%% file: TW-Feb-13-0323-R1.tex
\documentclass[journal]{IEEEtran}
\usepackage{amsmath}
\usepackage{amssymb}
\usepackage{amsfonts}
\usepackage{graphicx}
\usepackage{epsfig}
\usepackage{subfigure}
\usepackage{psfrag}
\usepackage{color}

\usepackage[noadjust]{cite}

\title{Opportunistic Wireless Energy Harvesting in Cognitive Radio Networks
\thanks{S. Lee and R. Zhang are with the Department of Electrical and Computer Engineering, National University of Singapore, Singapore (email: \{elelees, elezhang\}@nus.edu.sg). R. Zhang is also with the Institute for Infocomm Research, A*STAR, Singapore.} 
\thanks{K. Huang is with the Department of Applied Mathematics, Hong Kong Polytechnic University, Hong Kong (email: huangkb@ieee.org).}
}
\author{Seunghyun Lee, Rui Zhang,~\IEEEmembership{Member,~IEEE}, and Kaibin Huang,~\IEEEmembership{Member,~IEEE}}

\include{header}

\begin{document}
\maketitle \thispagestyle{empty}

\begin{abstract}
Wireless networks can be self-sustaining by harvesting energy from ambient \emph{radio-frequency} (RF) signals. Recently, researchers have made progress on designing efficient circuits and devices for RF energy harvesting suitable for low-power wireless  applications. Motivated by this and building upon the classic cognitive radio (CR) network model, this paper  proposes a novel method for wireless networks coexisting where low-power mobiles in a secondary network, called \emph{secondary transmitters} (STs), harvest ambient RF energy from transmissions by nearby active transmitters in a primary network, called \emph{primary transmitters} (PTs), while opportunistically accessing  the spectrum licensed to the primary network. We consider a stochastic-geometry model in which PTs and STs are distributed as independent homogeneous Poisson point processes (HPPPs) and communicate with their intended receivers at fixed distances. Each PT is associated with a \emph{guard zone} to protect its intended receiver from ST's interference, and at the same time delivers RF energy to STs located in its \emph{harvesting zone}. Based on the proposed model, we analyze the transmission probability of STs and the resulting spatial throughput of the secondary network. The optimal transmission  power and density of STs are derived for maximizing the secondary network throughput under the given outage-probability constraints in the two coexisting networks, which reveal key insights to the optimal network design. Finally, we show that our analytical result can be generally applied to a non-CR setup, where distributed wireless power chargers are deployed to power coexisting wireless transmitters in a sensor network.
\end{abstract}


\newtheorem{definition}{\underline{Definition}}[section]
\newtheorem{fact}{Fact}
\newtheorem{assumption}{Assumption}
\newtheorem{theorem}{\underline{Theorem}}[section]
\newtheorem{lemma}{\underline{Lemma}}[section]
\newtheorem{corollary}{\underline{Corollary}}[section]
\newtheorem{proposition}{\underline{Proposition}}[section]
\newtheorem{example}{\underline{Example}}[section]
\newtheorem{remark}{\underline{Remark}}[section]
\newtheorem{algorithm}{\underline{Algorithm}}[section]
\newcommand{\mv}[1]{\mbox{\boldmath{$ #1 $}}}

\section{Introduction} \label{Sec:Intro}
\PARstart{P}owering mobile devices by harvesting energy from ambient sources such as solar, wind, and kinetic activities  makes wireless networks not only environmentally friendly but also self-sustaining. Particularly, it has been reported in the recent literature that harvesting energy from ambient radio-frequency (RF) signals can power a network of low-power devices such as wireless sensors \cite{Zungeru:RFEnergyHarvesingManagementWSN:2012, Le:EfficientFarFiledRFEnergyHarvesting:2008, Vullers:MicropowerEnergyHarvesting:2009, Bouchouicha:AmbientRFEnergyHarvesting:2010, Paing:ResistorEmulation:2008, Jabbar:RFEnergyHarvestingMobile:2010}.  In theory, the maximum power available for RF energy harvesting at a free-space distance of $40$ meters is known to be $7$uW and $1$uW for 2.4GHz and 900MHz frequency, respectively \cite{Zungeru:RFEnergyHarvesingManagementWSN:2012}. Most recently, Zungeru \emph{et al.} have achieved harvested power of $3.5$mW at a distance of $0.6$ meter and $1$uW at a distance of $11$ meters using Powercast RF energy-harvester operating at $915$MHz \cite{Zungeru:RFEnergyHarvesingManagementWSN:2012}. It is expected that more advanced technologies for RF energy harvesting will be available in the near future due to e.g. the rapid advancement in designing highly efficient rectifying antennas \cite{Vullers:MicropowerEnergyHarvesting:2009}. 

In this work, we investigate the impact of RF energy harvesting on the newly emerging cognitive radio (CR) type of networks. To this end, we propose a novel method for wireless networks coexisting where transmitters from a secondary network, called \emph{secondary transmitters} (STs), either opportunistically harvest RF energy from transmissions by nearby transmitters from a primary network, or transmit signals if these \emph{primary transmitters} (PTs) are sufficiently far away. STs store harvested energy in rechargeable batteries with finite capacity and apply the available energy for subsequent transmissions when batteries are fully charged. The throughput of the secondary network is analyzed  based on a stochastic-geometry model, where the PTs and STs are distributed according to independent homogeneous Poisson point processes (HPPPs). In this model, each PT is assumed to randomly access the spectrum with a given probability and each active (transmitting) PT is centered at a \emph{guard zone} as well as a \emph{harvesting zone} that is inside the guard zone. As a result, each ST harvests energy if it lies in the harvesting zone of any active PT, or transmits if it is outside the guard zones of all active PTs, or is idle otherwise. This model is applied to maximize the spatial throughput of the secondary network by optimizing key parameters including the ST transmit power and density subject to given PT transmit power and density, guard/harvesting zone radius, and outage-probability constraints in both the primary and secondary networks.

Our work is motivated by a joint investigation of the proposed conventional \emph{opportunistic spectrum access} and the newly introduced \emph{opportunistic energy harvesting} in CR networks, i.e., during the idle time of STs due to the presence of nearby active PTs, they can take such an opportunity to harvest significant RF energy from primary transmissions. Specifically, as shown in Fig.~\ref{Fig:NetworkModel}, each ST can be in one of the following three modes at any given time:  \emph{harvesting mode} if it is inside the harvesting zone of an active PT and not fully charged;  \emph{transmitting mode} if it is fully charged and outside the guard zone of all active PTs; and \emph{idle mode} if it is fully charged but inside any of the guard zones, or neither fully charged nor inside any of the harvesting zones.

{\begin{figure} 
\centering
\includegraphics[width=9cm]{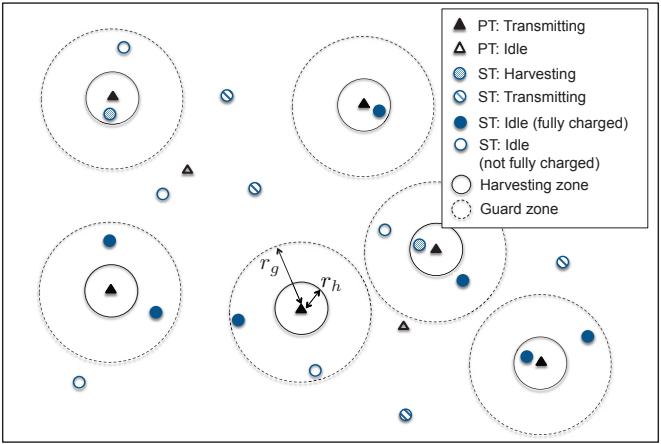} \vspace{1pt}
\caption{A wireless energy harvesting CR network in which PTs and STs are distributed as independent  HPPPs. Each PT/ST has its intended information receiver at fixed distances (not shown in the figure for brevity). ST harvests energy from a nearby PT if it is inside its  harvesting zone. To protect the primary transmissions, ST inside a guard zone is prohibited from transmission.}
\label{Fig:NetworkModel}
\end{figure}

\subsection{Related Work}
Recently, wireless communication powered by energy harvesting has emerged to be a new and active research area. However, due to energy harvesting, existing transmission algorithms for conventional wireless systems with constant power supplies (e.g., batteries) need to be redesigned to account for the new challenges such as random energy arrivals. For point-to-point wireless systems powered by energy harvesting, the optimal power-allocation algorithms have been designed and shown to follow modified water-filling by Ho and Zhang \cite{Ho&Zhang:OptimalEnergyAllocation:2012} and Ozel \emph {et al.} \cite{Ozel:EnergyHarvesting:2011}.  From a network perspective, Huang investigated the throughput of a mobile ad-hoc network (MANET) powered by energy harvesting  where the network spatial throughput is maximized by optimizing the transmit power level under an outage constraint \cite{Huang:ThroughputAdHocEnergyHarvesting:2012}. Furthermore, the performance of solar-powered wireless sensor/mesh networks has been analyzed in \cite{NiyatoHossainFallahi:Solar:2007}, in which various sleep and wakeup strategies are considered.

Among other energy scavenging sources such as solar and wind, background RF signals can be a viable new source for \emph{wireless energy harvesting} \cite{Brown:WirelessPowerTransfer:1984}. A new research trend on wireless power transfer is to integrate this technology with wireless communication. In \cite{Varshney:InformationEnergyTransfer:2008} and \cite{Grover:InformationEnergyTransfer:2010}, simultaneous wireless power and information transfer has been investigated, aiming at maximizing information rate and transferred power over single-antenna additive white Gaussian noise (AWGN) channels.  For broadcast channels, Zhang and Ho have studied multi-antenna transmission for simultaneous wireless information and power transfer with practical receiver designs such as time switching and power splitting \cite{Zhang:MIMOBroadcastingPowerTransfer}. Moreover, Zhou \emph{et al.} have proposed a new receiver design for enabling wireless information and power transmission at the same time, by judiciously integrating conventional information and energy receivers \cite{Zhou:WirelessPowerArchitectureDesign}. 
For point-to-point wireless systems, Liu \emph{et al.} have studied ``opportunistic'' RF energy harvesting where the receiver opportunistically harvests RF energy or decodes information subject to time-varying co-channel interference \cite{Liu:OpportunisticEnergyHarvesting}. More recently, Huang and Lau have proposed a new cellular network architecture consisting of power beacons deployed to deliver wireless energy to mobile terminals  and characterized the trade-off between the power-beacon density and cellular network spatial throughput \cite{Huang:WirelessPowerCellular}. 

In another track, the emerging CR technology enables efficient spectrum usage by allowing a secondary network to share the spectrum licensed to a primary network without significantly degrading its performance 
\cite{Haykin:CognitiveRadio:2005}. Besides active development of  algorithms
for opportunistic transmissions by secondary users (see e.g. \cite{QingZhao:DynamicSpectrumAccess:2007, RuiZhang:DynamicResourceAllocation:2010} and references therein), notable research has been pursued on characterizing the throughput of coexisting wireless networks based on the tool of stochastic geometry. For example, the capacity trade-offs between two or more coexisting networks sharing a common spectrum have been studied in \cite{Yin:ScalingLaws:2010, Huang:SpectrumSharing:2009, Lee:SpectrumSharingTxCapacity:2011}. Moreover, the outage probability of a Poisson-distributed CR network with guard zones  has been analyzed by Lee and Haenggi  \cite{Lee:PoissonCognitiveNetworks:2012}, where the secondary users opportunistically access the primary users' channel only when they are not inside any of the guard zones.

\subsection{Summary and Organization}
In this paper, we consider a CR network with time slotted transmissions and PT/ST locations modeled by independent HPPPs. The ST transmission power is assumed to be sufficiently small to meet the low-power requirement with RF energy harvesting. Under this setup, the main results of this paper are summarized as follows:

\begin{enumerate}

\item 
We propose a new CR network architecture where STs are powered by harvesting RF energy from active primary transmissions. We study the ST transmission probability as a function of ST transmit power in the presence of both guard zones and harvesting zones based on a Markov chain model. For the cases of single-slot and double-slot charging, we obtain the expressions of the exact ST transmission probability, while for the general case of multi-slot charging with more than two slots, we obtain the upper and lower bounds on the ST transmission probability. 

\item
With the result of ST transmission probability, we derive the outage probabilities of coexisting primary and secondary networks subject to their mutual interferences, based on stochastic geometry and a simplified assumption on the HPPP of transmitting STs with an effective density equal to the product of the ST transmission probability and the ST density. Furthermore, we maximize the spatial throughput of the secondary network under given outage constraints for the coexisting networks by jointly optimizing the ST transmission power and density, and obtain simple closed-form expressions of the optimal solution.

\item 
Furthermore, we show that our analytical result can be generally applied to even non-CR setups, where distributed wireless power chargers (WPCs) are deployed to power coexisting wireless information transmitters (WITs) in a sensor network, as shown in Fig.~\ref{Fig:NetworkModel_wo}. Practically, WPCs can be implemented as e.g. \emph{wireless charging vehicles} \cite{Xie:SensorNetworksImmortal:2012}, or fixed \emph{power beacons} \cite{Huang:WirelessPowerCellular} randomly deployed in a wireless sensor network. Based on our result for the CR network setup, we derive the maximum network throughput of such wireless powered sensor networks in terms of the optimal density and transmit power of WITs.

\end{enumerate}
{\begin{figure} 
\centering
\includegraphics[width=9cm]{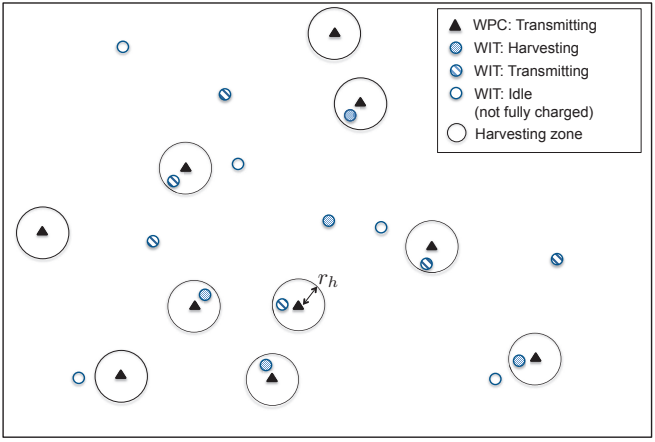} \vspace{1pt}
\caption{A wireless powered sensor network in which WPCs and WITs are distributed as independent HPPPs. Each WIT has intended receiver at a fixed distance (not shown in the figure for brevity). WIT harvests energy from a nearby WPC if inside its harvesting zone. Unlike the CR setup in Fig.~\ref{Fig:NetworkModel}, the guard zone is not applicable in this case, and thus a fully charged WIT can transmit at any time.}
\label{Fig:NetworkModel_wo}
\end{figure}

The remainder of this paper is organized as follows. Section~\ref{Sec:SysMod} describes the system model and performance metric. Section~\ref{Sec:TxProb} analyzes the transmission probability of energy-harvesting STs. Section~\ref{Sec:OutageProb} studies the outage probabilities in the primary and secondary networks. Section~\ref{Sec:NetThroughput} investigates the maximization of the secondary network throughput subject to the primary and secondary outage probability constraints. Section~\ref{Sec:Application} extends the result to the wireless powered sensor network setup. Finally, Section~\ref{Sec:Con} concludes the paper.

\section{System Model} \label{Sec:SysMod}

\subsection{Network Model}
As shown in Fig.~\ref{Fig:NetworkModel}, we consider a CR network in which  PTs and STs are distributed as independent HPPPs \footnote{ In general, transmitters' locations in cognitive radio networks may have non-homogeneous or even non-Poisson spatial distributions, which are difficult to characterize and not amenable to analysis. In this paper, we assume HPPP for transmitters' locations to obtain tractable analysis for the network performance. } with density $\lambda'_p$ and $\lambda_s$, respectively, with $\lambda'_p \ll \lambda_s$. It is assumed that time is slotted and each PT independently accesses the spectrum with probability $p$ at each time slot. Thus, the point process of active PTs forms another  HPPP with density $\lambda_p = p\lambda'_p$, according to the Coloring Theorem \cite{PoissonProcesses}, which varies independently over different slots. For convenience, we refer to active PTs simply as PTs in the rest of this paper. We denote the point processes of PTs and STs as $\Phi_p = \{X\}$ and $\Phi_s = \{Y\}$, respectively, where $X,Y \in \mR^2$ denote the coordinates of the PTs and STs, respectively. In addition, it is assumed that each PT/ST transmits with fixed power to its intended primary/secondary receiver (PR/SR) at distances $d_p$ and $d_s$, respectively, in random directions. We denote the fixed transmission power levels of PTs and STs as $P_p$ and $P_s$, respectively. We assume $P_p \gg P_s$ in this paper for energy harvesting applications of practical interest. 

STs access the spectrum of the primary network and thus their transmissions potentially interfere with PRs. To protect the primary transmissions,  STs are prevented from transmitting when they lie in any of the \emph{guard zones}, modeled as disks with a fixed radius centered at each PT. Specifically, let $b(T,x)\subset\mR^2$ represent a disk of radius $x$ centered at $T \in \mR^2$; then $b(X, r_g)$ denotes the guard zone with radius $r_g$ for protecting PT $X \in \Phi_p$. Define $\cG = \bigcup_{X \in \Phi_p}b(X,r_g)$ as the union of all PTs' guard zones; accordingly, an ST $Y \in \Phi_s$ cannot transmit if $Y \in \cG$. Note that in practice  the guard zone is usually centered at a PR rather than a PT as we have assumed, while our assumption is made to simplify our  analysis, similarly as in \cite{QingZhao:DynamicSpectrumAccess:2007}. We further assume $d_p \ll r_g$ to guarantee that guard zones centered at PTs (rather than PRs) will protect the primary transmissions properly. Under the above assumptions, the probability $p_g$ that a typical ST, denoted by $Y^\star$, does not lie in $\cG$ is equal to the probability that there is no PT inside the disk centered at $Y^\star$ with radius $r_g$, i.e., $b(Y^\star,r_g)$. Note that the number of PTs inside $b(Y^\star,r_g)$, denoted by $N$, is a Poisson random variable with mean $\pi r_g^2 \lambda_p$; thus, its probability mass function (PMF) is given by
\begin{equation} \label{Eq:PMF}
\Pr\{N=n\} = e^{-\pi r_g^2 \lambda_p}\frac{(\pi r_g^2 \lambda_p)^n}{n!}, \quad n=0,1,2,...
\end{equation}
Consequently, $p_g$ can be obtained as 
\begin{align}
p_g &= \Pr\{Y^\star \notin \cG\}  \label{Eq:GuardZoneDef}  \\
&= \Pr\{N = 0\} \label{Eq:GuardZoneRe} \\
&= e^{-\pi r_g^2 \lambda_p}. \label{Eq:GuardZoneRe:a}
\end{align}

We assume flat-fading channels with path-loss and Rayleigh fading; hence, the channel gains are modeled as exponential random variables. As a result, in a particular time slot, the signals transmitted from a PT/ST are received at the origin with power $g_X P_p |X|^{-\alpha}$ and $g_Y P_s |Y|^{-\alpha}$, respectively, where $\{g_X\}_{X\in\Phi_p}$ and $\{g_Y\}_{Y\in\Phi_s}$ are independent and identically distributed (i.i.d.) exponential random variables with unit mean, $\alpha >2$ is the path-loss exponent, and $|X|,|Y|$ denote the distances from node $X,Y$ to the origin, respectively.

\subsection{Energy-Harvesting Model} \label{Subsection:EHModel}
To make use of the RF energy as an energy-harvesting source, each RF energy harvester in an ST must be equipped with a  \emph{power conversion circuit} that can extract DC power from the  received electromagnetic waves \cite{Le:EfficientFarFiledRFEnergyHarvesting:2008}. Such circuits in practice have certain sensitivity requirements, i.e., the input power needs to be larger than a predesigned threshold for the circuit to harvest RF energy efficiently. This fact thus motivates us to define the \emph{harvesting zone}, which is a disk with radius $r_h$ centered at each PT $X\in \Phi_p$ with $r_h \ll r_g$. The radius $r_h$ is determined by the energy harvesting circuit sensitivity for a given $P_p$, such that only STs inside a harvesting zone can receive power larger than the energy harvesting threshold, which is given by $P_p r_h^{-\alpha}$. The power received by an ST outside any harvesting zone is too small to activate the energy harvesting circuit, and thus is assumed to be negligible in this paper.

Let $b(X, r_h)$ represent the harvesting zone centered at PT $X\in\Phi_p$ such that an ST $Y$ can harvest energy from one or more PTs if $Y \in \cH$, where $\cH = \bigcup_{X\in \Phi_p}b(X,r_h)$ denotes the union of the harvesting zones of all PTs. The probability $p_h$ that a typical ST $Y^\star$ lies in $\cH$ is equal to the probability that there is at least one PT inside the disk $b(Y^\star,r_h)$. Similar to \eqref{Eq:PMF}, the number of PTs inside $b(Y^\star,r_h)$, denoted by $K$, is a Poisson random variable with mean $\pi r_h^2 \lambda_p$ and PMF given by
\begin{equation}
\Pr\{K=k\} = e^{-\pi r_h^2 \lambda_p}\frac{(\pi r_h^2 \lambda_p)^k}{k!}, \quad k=0,1,2,...
\end{equation}
Accordingly, $p_h$  is given by
\begin{align}
p_h & = \Pr\{Y^\star \in \cH\} \\
& = \Pr\{K \geq 1\} \label{Eq:HarvestingZoneDef}\\
& = \sum_{k=1}^{\infty} e^{-\pi r_h^2 \lambda_p}\frac{(\pi r_h^2 \lambda_p)^k}{k!} \label{Eq:p_hSum}\\
& = 1-e^{-\pi r_h^2 \lambda_p} \label{Eq:p_h}.
\end{align}
Since $\lambda_p$ and $r_h$ are both practically small, we can assume $\pi r_h^2 \lambda_p \ll 1$. Thus, $p_h$ given in \eqref{Eq:p_hSum} can be approximated as $\Pr\{K=1\}$ by ignoring the higher-order terms with $k>1$. Therefore,  when $Y^\star \in \cH$, $Y^\star$ is inside the harvesting zone of one single PT most probably, which equivalently means that the harvesting zones of different PTs do not overlap at most time. As a result, the amount of average power harvested by $Y^\star \in \cH$ in a time slot can be lower-bounded by $\eta P_p R^{-\alpha}$ where $R \leq r_h$ denotes the distance between $Y^\star$ and its nearest PT, and $0<\eta < 1$ denotes the harvesting efficiency. Note that the harvested power has been averaged over the channel short-term fading within a slot. 

\subsection{ST Transmission Model}
We assume that each ST has a battery of finite capacity equal to the minimum energy required for one-slot transmission with power $P_s$ for simplicity. Upon the battery being fully charged, an ST will transmit with all stored energy in the next slot if it is outside all the guard zones. We denote the probability that $Y^\star$ has been fully charged at the beginning of a time slot as $p_f$ and the probability that it will be able to transmit in this slot as $p_t$. As mentioned above, the point process of PTs $\Phi_p$ varies independently over different slots, and thus the events that an ST has been fully charged in one slot and that it is outside all the guard zones in the next slot are independent. Consequently, $p_t$ can simply be obtained as
\begin{equation} \label{Eq:TxProb:Def}
p_t = p_f p_g,
\end{equation}
where $p_g$ is given in \eqref{Eq:GuardZoneRe:a}, and $p_f$ will be derived in Section~\ref{Sec:TxProb}.

\subsection{Performance Metric}
For both PRs and SRs, the received signal-to-interference-plus-noise ratio (SINR) is required to exceed a given target for reliable transmission. Let $\theta_p$ and $\theta_s$ be the target SINR for the PR and SR, respectively. The outage probability  is then defined as $\Pout^{(p)}=\Pr\{\SINR^{(p)}<\theta_p\}$ for the primary network  and  $\Pout^{(s)}=\Pr\{\SINR^{(s)}<\theta_s\}$ for the secondary network. The outage-probability constraints are applied such that  $\Pout^{(p)}\leq \epsilon_p$ and $\Pout^{(s)}\leq \epsilon_s$ with given $0 < \epsilon_p, \epsilon_s <1$. Note that the transmitting STs in general do not form an HPPP due to the presence of guard zones and energy harvesting zones, but their average density over the network is given by $p_t\lambda_s$. Accordingly, given fixed PT density $\lambda_p$ and transmission power $P_p$, the performance metric of the secondary network is the spatial throughput $\cC_s$ (bps/Hz/unit-area) given by
\begin{equation} 
\cC_s = p_t \lambda_s \log_2(1+\theta_s), \label{Eq:NetThroughput}
\end{equation} 
under the given primary/secondary outage probability constraints $\epsilon_p$ and $\epsilon_s$.

\section{Transmission Probability of Secondary Transmitters} \label{Sec:TxProb}

In this section, the transmission probability of a typical ST $p_t$ given in \eqref{Eq:TxProb:Def} is analyzed using the Markov chain model. For convenience, we define $M$ as the maximum number of energy-harvesting time slots required to fully charge the  battery of an ST. Since the minimum power harvested by an ST in one slot is $\eta P_p r_h^{-\alpha}$, which occurs when the ST is at the edge of a harvesting zone, it follows that $M = \l\lceil\frac{P_s}{\eta P_p r_h^{-\alpha}}\r\rceil$, where $\lceil x \rceil$ denotes the smallest integer larger than or equal to $x$. Note that  $M=1$ corresponds to the case where the battery is fully charged within one slot time; thus this case is  referred to as \emph{single-slot charging}. Similarly, the case of $M=2$ is referred to as \emph{double-slot charging}. It will be shown in this section that if $M=1$ or $M=2$, the battery power level can be exactly modeled by a finite-state Markov chain; hence, the transmission probability $p_t$ can be obtained. However, for \emph{multi-slot charging} with $M>2$, only upper and lower bounds on $p_t$ are obtained based on the Markov chain analysis for the case of $M=2$.

\subsection{Single-Slot Charging ($M=1$)}
If $0< P_s \leq \eta P_p r_h^{-\alpha}$, the battery of an ST is fully charged within a slot, i.e., $M=1$. It thus follows that the battery power level can only be either 0 or $P_s$ at the beginning of each slot. Consider the finite-state Markov chain with state space $\{0,1\}$ with states $0$ and $1$ denoting the battery level of power $0$ and $P_s$, respectively. Furthermore, let $\mathbf{P}_1$ represent the state-transition probability matrix that can be obtained as
\begin{equation} \label{Eq:Matrix:Case1}
\mathbf{P}_1=\l[\begin{array}{cc}
1-p_h & p_h \\ p_g & 1-p_g  \end{array}\r]
\end{equation}
with $p_g$ and $p_h$ given in \eqref{Eq:GuardZoneRe:a} and \eqref{Eq:p_h}, respectively. Then $p_t$ can be obtained by finding the steady-state probability of the assumed Markov chain, as given in the following proposition. 
\begin{proposition} \label{Prop:TxProb:Case1}
If $0<P_s \leq \eta P_p r_h^{-\alpha}$ or $M=1$ (single-slot charging), the transmission probability of a typical ST is given by
\begin{align} 
p_t &= \frac{p_h}{p_h+p_g} p_g \label{Eq:TxProb:Case1_2}\\
\label{Eq:TxProb:Case1} & = \frac{(1-e^{-\pi r_h^2 \lambda_p})e^{-\pi r_g^2 \lambda_p}}{1-e^{-\pi r_h^2 \lambda_p} + e^{-\pi r_g^2 \lambda_p}}.
\end{align}
\end{proposition}

\begin{proof}
Let the steady-state probability of the two-state Markov chain be denoted by $\boldsymbol{\pi}_1=[\pi_{1,0},\pi_{1,1}]$, where  $\boldsymbol{\pi_1}$ is the left eigenvector of $\mathbf{P}_1$ corresponding to the unit eigenvalue such that
\begin{equation} \label{Eq:SteadyState}
\boldsymbol{\pi}_1\mathbf{P}_1=\boldsymbol{\pi}_1.
\end{equation}
From \eqref{Eq:SteadyState}, the steady-state distribution of the battery power level at a typical ST is obtained as 
\begin{equation}
\pi_{1,0} = \frac{p_g}{p_h+p_g}, \quad \pi_{1,1} = \frac{p_h}{p_h+p_g}.
\end{equation}
Note that the probability that an ST is fully charged at the beginning of each slot as defined in \eqref{Eq:TxProb:Def} is $p_f = \pi_{1,1}$ in this case. Consequently, from \eqref{Eq:TxProb:Def}, the desired result in \eqref{Eq:TxProb:Case1_2} is obtained.
\end{proof}

It is observed from \eqref{Eq:TxProb:Case1} that in the single-slot charging case, $p_t$ depends only on $\lambda_p$, $r_h$ and $r_g$, but is not related to $P_s$. The reason is that the battery of an ST is guaranteed to be fully charged over one slot if it gets into a harvesting zone; hence, the probability that an ST is fully charged $p_f = \pi_{1,1} = \frac{p_h}{p_h+p_g}$ does not depend on $P_s$. 

\subsection{Double-Slot Charging ($M=2$)}
\label{Subsection:Case2}
If $\eta P_p r_h^{-\alpha} < P_s \leq 2\eta P_p r_h^{-\alpha}$ or $M=2$, an ST needs at most 2 slots of harvesting to make the battery fully charged. To establish the Markov chain model for this case, we divide the harvesting zone $b(X, r_h)$ into two disjoint regions, $b(X, h_1)$ and $a(X, h_1, r_h)$, where $h_1 = \l(\frac{P_s}{\eta P_p}\r)^{-\frac{1}{\alpha}}<r_h$ and $a(T, x, y) = b(T,y) \backslash b(T,x)$ denotes the annulus with radii $0<x<y$ centered at $T\in \mR^2$.  It then follows that the region $b(X, h_1)$ consists of the locations at which the power harvested by a typical ST $Y^\star$ from PT $X$ is greater than or equal to $P_s$ (i.e., single-slot charging region), while the region $a(X, h_1, r_h)$ corresponds to the locations at which the power harvested by $Y^\star$ is greater than or equal to $\frac{1}{2}P_s$ but smaller than $P_s$ (see Fig.~\ref{Fig:Division1}).  For convenience, we define $\cH_1 = \bigcup_{X\in\Phi_p}b(X, h_1)$ and $\cH_2 = \bigcup_{X\in\Phi_p}a(X, h_1, r_h)$. Note that $\cH = \cH_1 \cup \cH_2$. We reasonably assume that $\cH_1$ and $\cH_2$ are disjoint  since the harvesting zones are most likely disjoint as mentioned in Section~\ref{Subsection:EHModel}.

\begin{figure}
\centering
\includegraphics[width=7cm]{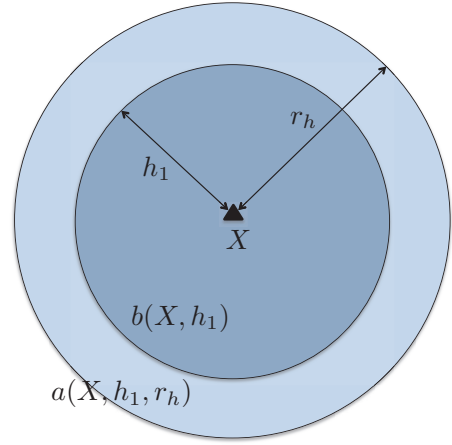}\vspace{5pt} 
\caption{Divided harvesting zone for the case of double-slot charging ($M=2$).}
 \label{Fig:Division1}
\end{figure}

Consider a 3-state Markov chain with state space $\{0,1,2\}$. Since the battery power level can only be either $0$ or in the range $[\frac{1}{2}P_s, P_s]$ since $\eta P_p r_h^{-\alpha} \geq \frac{1}{2}P_s$ in this case, we define state $0$ as the battery level of power 0, state $1$ with the power level in the range $[\frac{1}{2}P_s, P_s)$, and state $2$ with the power level equal to $P_s$. Note that in order to transit from state $0$ to $1$, $0$ to $2$, and $1$ to $2$, the harvested power at $Y^\star$ needs to be $\frac{1}{2}P_s \leq \eta P_p R^{-\alpha} < P_s$, $\eta P_p R^{-\alpha} \geq P_s$, and $\eta P_p R^{-\alpha} \geq \frac{1}{2}P_s$, respectively (or equivalently $Y^\star$ needs to be inside $\cH_2$, $\cH_1$, and $\cH$, respectively). Thanks to the fact that the minimum charging power is always larger than or equal to $\frac{1}{2}P_s$ in this case, we can determine the probability of the transition from state $1$ to $2$, i.e., from the battery power level in the range of $[\frac{1}{2}P_s,P_s)$ to $P_s$, which occurs when $Y^\star$ is (anywhere) inside a harvesting zone (see Fig.~\ref{Fig:Battery}). Accordingly, the state-transition probability matrix for the assumed 3-state Markov chain (see Fig.~\ref{Fig:MarkovChain}) is given as 
\begin{equation} \label{Eq:Matrix:Case2}
\mathbf{P}_2=\l[\begin{array}{ccc}
1-p_h & p_{2} & p_{1}\\0 & 1-p_h & p_h  \\ p_g & 0 & 1-p_g\end{array} \r],
\end{equation}
where $p_1 = \Pr\{Y^\star \in \cH_1\}$ and $p_2 = \Pr\{Y^\star\in \cH_2\}$. Notice that $p_1 + p_2 = p_h = 1-e^{-\pi r_h^2 \lambda_p}$, since $\cH_1 \cup \cH_2 = \cH$ and we have assumed that $\cH_1$ and $\cH_2$ are disjoint sets.  
Similarly to \eqref{Eq:HarvestingZoneDef}, the probability $p_1$ is given as
\begin{align}
p_1 & = \Pr\{Y^\star \in \cH_1\} \\
& = 1-e^{-\pi h_1^2 \lambda_p} \label{Eq:p_1},
\end{align}
and $p_2$ is given as
\begin{align}
p_2 & = p_h - p_1 \\
& = e^{-\pi h_1^2 \lambda_p} - e^{-\pi r_h^2 \lambda_p} \label{Eq:p_2}.
\end{align}
Then we can obtain $p_t$ for this case as given in the following proposition.

\begin{figure}
\centering
\subfigure[Battery power state of ST]{
\centering
\includegraphics[width=7cm]{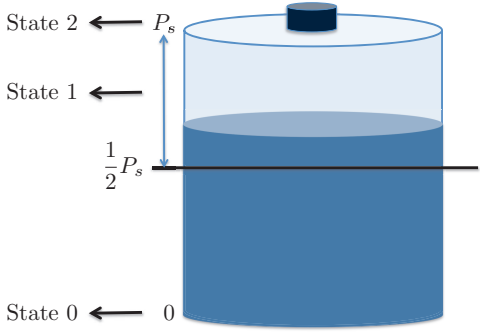}\vspace{5pt}\label{Fig:Battery}}
\subfigure[Markov chain model]{
\centering
\includegraphics[width=7cm]{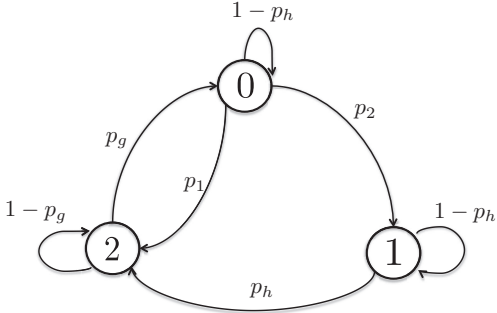}\vspace{5pt}\label{Fig:MarkovChain}} \vspace{-5pt}
\caption{The battery power state for the case of $M=2$ and the corresponding 3-state Markov chain model, where (a) shows an example of the ST being in state $1$ of the Markov model in (b), i.e., the current battery power level is in the range $[\frac{1}{2}P_s, P_s)$. }
\label{Fig:Case2}
\end{figure}

\begin{proposition} \label{Prop:TxProb:Case2}
If $\eta P_p r_h^{-\alpha} < P_s \leq 2\eta P_p  r_h^{-\alpha}$ or $M=2$ (double-slot charging), the transmission probability of a typical ST is given by
\begin{align} 
p_t & = \frac{p_h}{p_h+p_g\l(1+\frac{p_{2}}{p_h}\r)}p_g \label{Eq:TxProb:Case2_2}\\
 \label{Eq:TxProb:Case2} & = \frac{(1-e^{-\pi r_h^2 \lambda_p})e^{-\pi r_g^2 \lambda_p}}{1-e^{-\pi r_h^2 \lambda_p} + e^{-\pi r_g^2 \lambda_p}\l(1+\frac{e^{-\pi h_1^2 \lambda_p} - e^{-\pi r_h^2 \lambda_p}}{1-e^{-\pi r_h^2 \lambda_p}}\r)}.
\end{align}
\end{proposition}

\begin{proof}
The result in \eqref{Eq:TxProb:Case2_2} can be obtained by following the similar procedure as in the proof of Proposition~\ref{Prop:TxProb:Case1}, i.e., by solving $\boldsymbol{\pi}_2\mathbf{P}_2 = \boldsymbol{\pi}_2$, where $\boldsymbol{\pi}_2$ is the steady-state probability vector given by $\boldsymbol{\pi}_2=[\pi_{2,0}, \pi_{2,1}, \pi_{2,2}]$. Then, we obtain $p_f = \pi_{2,2}$ and then \eqref{Eq:TxProb:Case2_2} is obtained from \eqref{Eq:TxProb:Def}.    
\end{proof}
It is worth noting from \eqref{Eq:TxProb:Case2} that $p_t$ in this case is a decreasing function of $P_s$ since $h_1 = \l(\frac{P_s}{\eta P_p}\r)^{-\frac{1}{\alpha}}$ in \eqref{Eq:TxProb:Case2} is such a  function. In other words, if $P_s$ increases with fixed $P_p$ and $r_h$, then the size of $b(X,h_1) $ (single-slot charging region) becomes smaller, which results in an ST harvesting for two slots to be fully charged more frequently, and thus a smaller $p_f$. Hence, $p_t$ becomes smaller as well given $p_t = p_f p_g$ in \eqref{Eq:TxProb:Def}.

\subsection{Multi-Slot Charging ($M>2$)} \label{Subsection:bounds}
For multi-slot charging with $P_s > 2\eta P_p  r_h^{-\alpha}$ or $M>2$, the minimum charging power at the edge of the harvesting zone, $\eta P_p r_h^{-\alpha}$,  is smaller than $\frac{1}{2}P_s$. Unlike the previous two cases of $M=1$ and $M=2$, the battery power level in this case cannot be characterized exactly by a finite-state Markov chain since it is not possible in general to uniquely determine the state-transition probabilities.\footnote{For instance, if $M=3$, following the previous two cases, we may divide the battery power level into $4$ levels as $0$, $[\frac{1}{3}P_s,\frac{2}{3}P_s)$, $[\frac{2}{3}P_s,P_s)$, and $P_s$ and match each level to the states $0,1,2$, and $3$, respectively. Then it can be easily shown that the transition probabilities are unknown for some of the state transitions, e.g., from state $1$ to $2$.} 
However, we have shown that for the case of $M=2$, the battery power level can indeed be characterized with a 3-state Markov chain regardless of the fact that we do not know the exact value of the battery power level in state $1$, but rather only know its range $[\frac{1}{2}P_s,P_s)$, provided that the minimum charging power $\eta P_p r_h^{-\alpha}$ is no smaller than $\frac{1}{2}P_s$. Based on this result, we obtain both the upper  and lower bounds on $p_t$ for the case with $M>2$ as follows. 

\begin{figure}
\centering
\includegraphics[width=7cm]{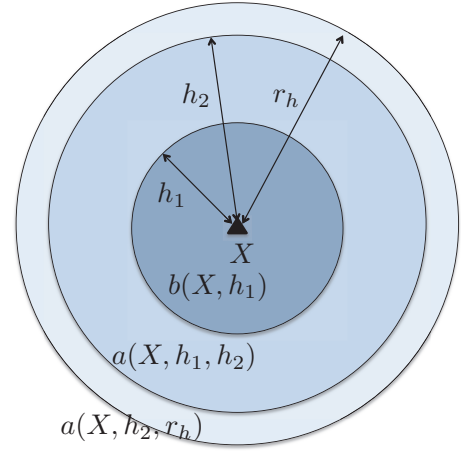}\vspace{5pt} 
\caption{Divided harvesting zone for the case of $M>2$. In this case, the amount of power harvested from PT $X$ in $a(X,h_2,r_h)$ is either overestimated as $\frac{1}{2}P_s$ or underestimated as $0$ to obtain an upper/lower bound on $p_t$ in Section~\ref{Subsection:bounds}.}
 \label{Fig:Division2}
\end{figure}

As shown in Fig.~\ref{Fig:Division2}, we divide the harvesting zone into $3$ disjoint regions $b(X, h_1)$, $a(X,h_1,h_2)$, and $a(X,h_2,r_h)$, where $0<h_1<h_2<r_h$ with $h_1$ given in the case of $M=2$  and $h_2 = \l(\frac{P_s}{2\eta P_p}\r)^{-\frac{1}{\alpha}}$. Note that $b(X, h_1)$ is also defined in the case of $M=2$, while the region $a(X, h_1, h_2)$ consists of the locations in $b(X,r_h)$ at which the power harvested from PT $X$ is larger than or equal to $\frac{1}{2}P_s$, but smaller than $P_s$, and the region $a(X, h_2, r_h)$ consists of the remaining locations in $b(X,r_h)$ at which the harvested power is smaller than $\frac{1}{2}P_s$. Then, if we assume that the power harvested from a PT in the region $a(X, h_2, r_h)$ is equal to $\frac{1}{2}P_s$ (an overestimation), we can obtain an upper bound on $p_t$; however, if we assume it is equal to $0$ (an underestimation), we can then obtain a lower bound on $p_t$, by applying a similar analysis over the 3-state Markov chain as for the case of $M=2$. For convenience, we define the following mutually exclusive sets $\cA_1 = \bigcup_{X\in\Phi_p}b(X, h_1)$, $\cA_2 = \bigcup_{X\in\Phi_p}a(X, h_1, h_2)$, and $\cA_3 = \bigcup_{X\in\Phi_p}a(X, h_2, r_h)$, where $\cA_1 = \cH_1$ and $\cA_1 \cup \cA_2 \cup \cA_3 = \cH$. Let $p'_2 = \Pr\{Y^\star \in \cA_2\}$ and $p_3 = \Pr\{Y^\star \in \cA_3\}$. It then follows that $p_1 + p'_2 + p_3 = p_h$, where $p_1$ is given in \eqref{Eq:p_1} and 
\begin{align}
p'_2 &= \Pr\{Y^\star \in \cA_1 \cup \cA_2\} - \Pr\{Y^\star \in \cA_1\} \nn\\
&= e^{-\pi \lambda_p h_1^2} - e^{-\pi \lambda_p h_2^2}, \label{Eq:p'_2} \\
p_3 &= p_h - p_1 - p'_2 =  e^{-\pi \lambda_p h_2^2} - e^{-\pi \lambda_p r_h^2}. \label{Eq:p_3}
\end{align}
The following proposition is then obtained.

\begin{proposition} \label{Theorem:Bound:TxProb}
If $P_s > 2\eta P_p  r_h^{-\alpha}$ or $M>2$, the transmission probability of an ST is bounded as
\begin{equation} \label{Bound:TxProb}
\frac{p_1+p'_2}{(p_1+p'_2) + p_g\l(1+\frac{p'_2}{p_1+p'_2}\r)}p_g < p_t < \frac{p_h}{p_h + p_g\l(1+\frac{p'_2 + p_3}{p_h}\r)}p_g.
\end{equation}
\end{proposition}
\begin{proof}
See Appendix~\ref{Proof:Theorem:Bound:TxProb}.
\end{proof}

It is worth mentioning that the upper bound on $p_t$ is a decreasing function of $P_s$ since $h_1 = \l(\frac{P_s}{\eta P_p}\r)^{-\frac{1}{\alpha}}$. Also note that the bounds in \eqref{Bound:TxProb} are tight in the case of $M=1$ or $M=2$, since $p'_2=p_3=0$ with $M=1$, and $p'_2 = p_2$ and $p_3=0$ with $M=2$, thus leading to the  same results in \eqref{Eq:TxProb:Case1_2} and \eqref{Eq:TxProb:Case2_2}, respectively.

Note that unlike the case of $M=2$, it is not possible to verify in general whether $p_t$ for the case of $M>2$ is a decreasing function of $P_s$ or not; however, it is conjectured to be so since a larger value of $P_s$ will generally render an ST spend more time to be fully charged. We verify this by simulation in the following subsection (see Fig.~\ref{Fig:TxProb}). 

\subsection{Numerical Example} \label{Subsec:Numerical}
To verify the results on $p_t$, we provide numerical examples as shown in Figs.~\ref{Fig:TxProb}, \ref{Fig:TxProbVSPTDen}, and \ref{Fig:TxProbVSr_g}. For all of these examples, we set the path-loss exponent as $\alpha=4$ and the harvesting efficiency as $\eta = 0.1$.

\begin{figure}
\centering
\includegraphics[width=9cm]{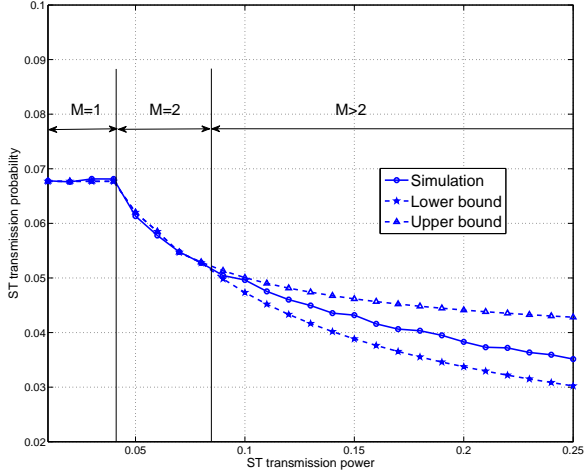}\vspace{5pt} 
\caption{ST transmission probability $p_t$ versus ST transmission power $P_s$, with $\lambda_p = 0.01$, $r_g = 4$, $r_h=1.5$, and $P_p = 2$.}
\label{Fig:TxProb}
\end{figure}

\begin{figure}
\centering
\includegraphics[width=9cm]{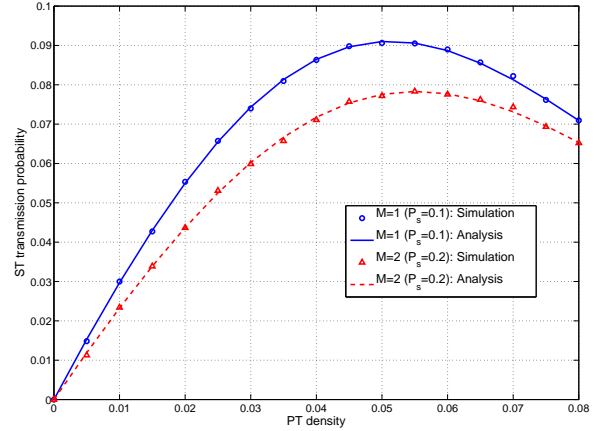}\vspace{5pt} 
\caption{ST transmission probability $p_t$ versus PT density $\lambda_p$, with $r_g=3$, $r_h=1$ and $P_p=1$.}
 \label{Fig:TxProbVSPTDen}
\end{figure}

\begin{figure}
\centering
\includegraphics[width=9cm]{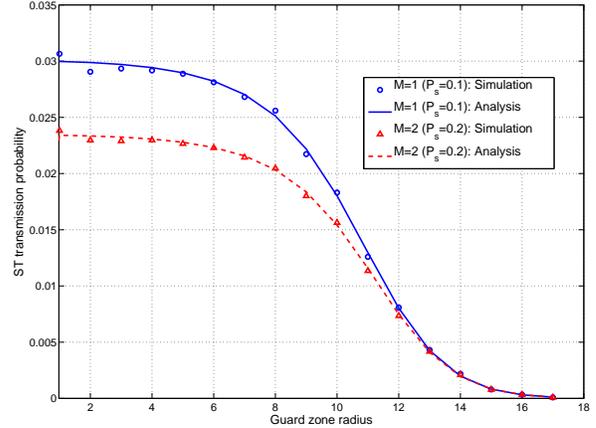}\vspace{5pt} 
\caption{ST transmission probability $p_t$ versus the radius of guard zone $r_g$, with $\lambda_p=0.01$, $r_h=1$,  and $P_p=1$.} 
\label{Fig:TxProbVSr_g}
\end{figure}

In Fig.~\ref{Fig:TxProb}, we show ST transmission probability $p_t$ versus ST transmission power $P_s$.  It is worth noting that $M=1$ if $0< P_s \leq \eta P_p r_h^{-\alpha}$, $M=2$ if $\eta P_p r_h^{-\alpha}< P_s \leq 2\eta P_p r_h^{-\alpha}$, and $M>2$ if $P_s > 2\eta P_p r_h^{-\alpha}$.  It is observed that $p_t$ is constant if $M=1$, but is a decreasing function of $P_s$ if $M=2$, which agrees with the results in \eqref{Eq:TxProb:Case1} and \eqref{Eq:TxProb:Case2}, respectively. It is also shown that if $M>2$, $p_t$ is still a decreasing function of $P_s$ as we conjectured. Moreover, the upper bound and lower bound on $p_t$ obtained in \eqref{Bound:TxProb} for $M>2$ are depicted in this figure. 
These bounds are observed to be tight when $M=1$ and $M=2$, while they get looser with increasing $P_s$ when $M>2$. The reason is that the size of the region $a(X,h_2,r_h)$ shown in Fig.~\ref{Fig:Division2}, in which we overestimate or underestimate the harvested power, enlarges with increasing $P_s$. However, since only small value of $P_s$ is of our interest, we can assume that these bounds are reasonably accurate for small values of $M$.

Fig.~\ref{Fig:TxProbVSPTDen} shows $p_t$ versus PT density $\lambda_p$. It is observed that for both $M=1$ and $M=2$, $p_t$ first increases with $\lambda_p$ when $\lambda_p$ is small but starts to  decrease with $\lambda_p$ when $\lambda_p$ becomes sufficiently large. This can be explained as follows.  If $\lambda_p$ is small, increasing $\lambda_p$ is more beneficial since each ST will get charged more frequently and thus be able to transmit (i.e., $p_f$ increases more substantially than the decrease of $p_g$). However, after $\lambda_p$ exceeds a certain threshold, increasing $\lambda_p$ will more pronounce the effect of guard zones and thus make STs transmit less frequently (i.e., $p_g$ decreases more substantially than the increase of $p_f$).

In Fig.~\ref{Fig:TxProbVSr_g}, we show $p_t$ versus the guard zone radius $r_g$. It is observed that $p_t$ is a decreasing function of $r_g$. Intuitively, this result is expected since larger $r_g$ results in STs transmitting less frequently, i.e., smaller values of $p_g$, and it is known from \eqref{Eq:TxProb:Def} that $p_t=p_f p_g$.

\section{Outage Probability} \label{Sec:OutageProb}
In this section, the outage probabilities of both the primary and secondary networks are studied. Let $\Phi_t$ denote the point process of the active (transmitting) STs. In addition, let $I_p$ and $I_s$ indicate the aggregate interference at the origin from all PTs and active STs, respectively, which are modeled by \emph{shot-noise processes} \cite{PoissonProcesses}, given by $I_{p}=\sum_{X\in\Phi_p}g_X P_p|X|^{-\alpha}$ and $I_{s}=\sum_{Y\in\Phi_t}g_Y P_s|Y|^{-\alpha}$, respectively. Note that in general, due to the presence of the guard zone and/or harvesting zone, in each time slot, the point process  $\Phi_t$ is not necessarily an HPPP; thus,  $I_s$ is not the shot-noise process of an HPPP. Accordingly, the outage probabilities $\Pout^{(p)}$ and $\Pout^{(s)}$ for primary and secondary networks, both related to $I_s$, are difficult to be characterized exactly. To overcome this difficulty, we make the following assumption on the process of active STs.

\begin{assumption} \label{Assumption:Approximation}
The point process of active STs $\Phi_t$ is an HPPP with density $p_t\lambda_s$.
\end{assumption} 
 It is shown in Fig.~\ref{Fig:CDF} that the cumulative distribution function (CDF) of $I_s$, given by $\Pr\{I_s \leq x\}$, obtained by simulations, can be well approximated by that of approximated $I_s$ based on Assumption~\ref{Assumption:Approximation}. Further verifications of Assumption~\ref{Assumption:Approximation} will be given later by simulations (see Figs.~\ref{Fig:OutageProb} and \ref{Fig:OutageProb_vs_P_s}).

\begin{figure}
\centering
\includegraphics[width=9cm]{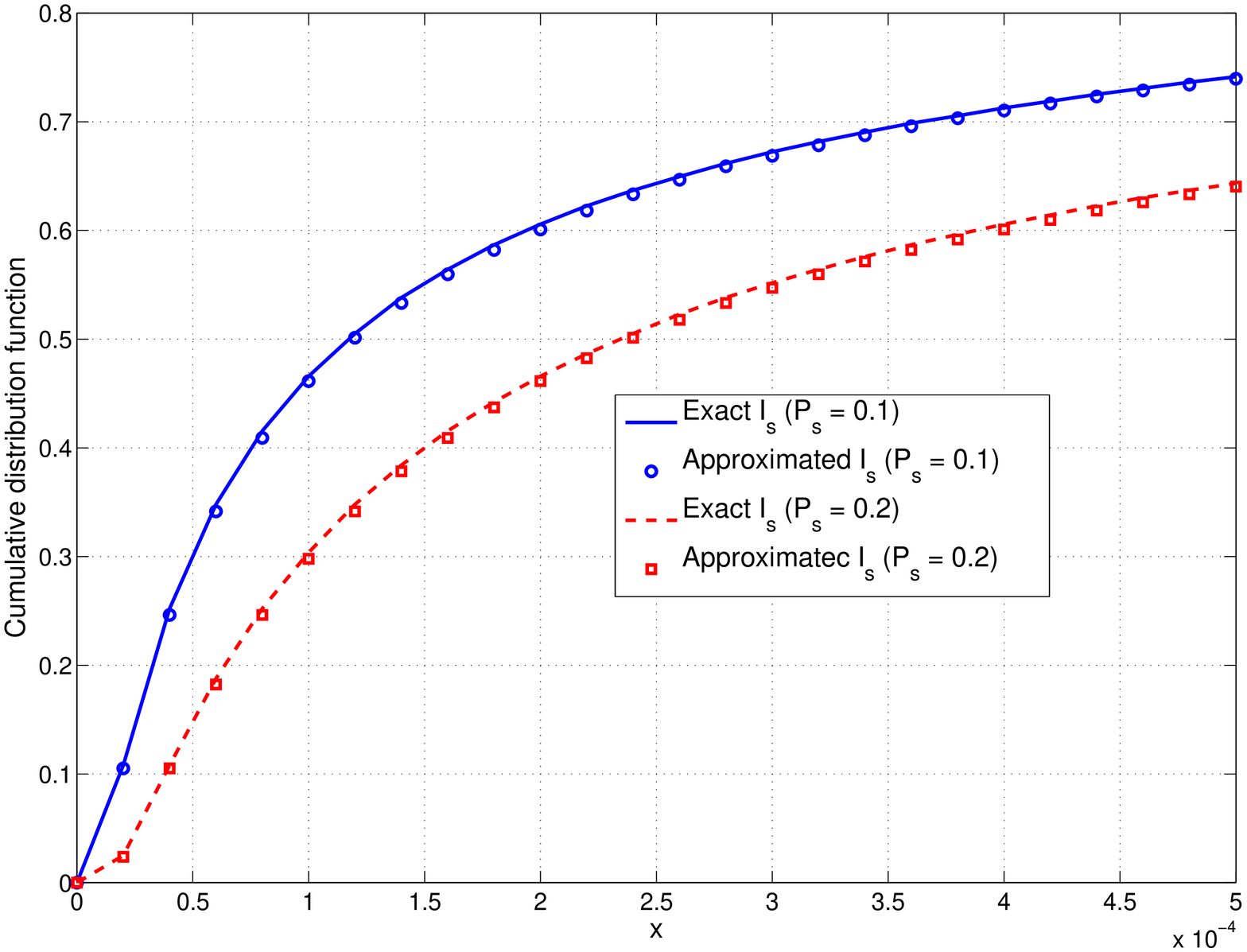}\vspace{5pt} 
\caption{The CDF of exact $I_s$ and approximated $I_s$ (based on Assumption~\ref{Assumption:Approximation}) with $\alpha=4$, $\eta = 0.1$, $r_g = 3$, $r_h=1$, $\lambda_s = 0.2$, $\lambda_p = 0.01$, and $P_p = 2$.} 
\label{Fig:CDF}
\end{figure}

Let $\Lambda(\lambda)$ denote the HPPP with density $\lambda >0$. Under Assumption~\ref{Assumption:Approximation}, the distribution of $\Phi_t$ is the same as that of $\Lambda(p_t\lambda_s)$. It thus follows that $I_s$ can be rewritten as
\begin{equation} \label{Eq:Is_Approximation}
I_s = \sum_{Y\in\Lambda(p_t\lambda_s)}g_Y P_s|Y|^{-\alpha}.
\end{equation}

Consider first the outage probability of the primary network, $\Pout^{(p)}$, which can be characterized by considering a typical PR located at the origin. Slivnyak's theorem \cite{PoissonProcesses} states that an additional PT corresponding to the PR at the origin does not affect the distribution of $\Phi_{p}$. Therefore, the outage probability of the PR at the origin is expressed as 
\begin{equation} \label{Eq:PoutPrimary}
\Pout^{(p)}  = \Pr \l\{\frac{g_p P_p d_p^{-\alpha}}{I_p + I_s + \sigma^2} < \theta_p  \r\},
\end{equation} 
where $g_p$ is the channel power between the  PR at the origin and its corresponding PT, and $\sigma^2 $ is the AWGN  power. Then, $\Pout^{(p)}$ is obtained in the following lemma. 

\begin{lemma} \label{Lemma:Approx:PoutPrimary2}
Under Assumption~\ref{Assumption:Approximation}, the outage probability of a typical PR at the origin is given by
\begin{equation} \label{Approx:PoutPrimary2}
\Pout^{(p)} = 1-\exp \l(-\tau_p \r), 
\end{equation}
where
\begin{equation} \label{Eq:Tau_p}
\tau_p = \l(\lambda_p + p_t\lambda_s\l(\frac{P_s}{P_p}\r)^{\frac{2}{\alpha}}\r)\theta_p^{\frac{2}{\alpha}}d_p^2\varphi + \frac{\theta_p d_p^\alpha \sigma^2}{P_p},
\end{equation}
$\varphi = \pi\frac{2}{\alpha}\Gamma(\frac{2}{\alpha})\Gamma(1-\frac{2}{\alpha})$,  with $\Gamma(x) = \int_0^\infty y^{x-1}e^{-y}dy$ denoting the Gamma function.
\end{lemma}
\begin{proof} 
See Appendix~\ref{Proof:Lemma:Approx:PoutPrimary2}.
\end{proof}

\begin{figure} 
\centering
\includegraphics[width=9cm]{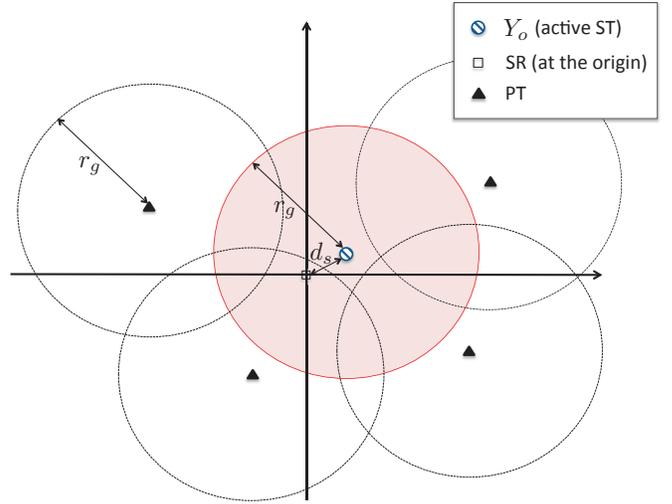} \vspace{5pt}
\caption{A typical SR located at the origin, for which there is no PT inside the shaded region $b(Y_o,r_g)$.}
\label{Fig:OutageSecondary}
\end{figure} 

Next, consider the outage probability of the secondary network, $\Pout^{(s)}$, which can be characterized by a typical SR located at the origin. Note that there must be an active ST, denoted by $Y_o$, corresponding to the SR at the origin. Since an ST cannot transmit if it is inside any guard zone, to accurately approximate $\Pout^{(s)}$ under Assumption~\ref{Assumption:Approximation}, we consider the outage probability conditioned on that $Y_o$ is outside all the guard zones and thus there is no PT inside the disk of radius $r_g$ centered at $Y_o$ (see Fig.~\ref{Fig:OutageSecondary}). Let the event in the above condition be denoted by $\cE = \{\Phi_p \cap b(Y_o,r_g) = \emptyset\}$. Then the outage probability of a typical SR at the origin can be obtained as
\begin{equation} \label{Eq:PoutSecondary}
\Pout^{(s)} = \Pr\l\{\frac{g_s P_s d_s^{-\alpha}}{I_p+I_{s} + \sigma^2} <\theta_s \l| \cE \r.\r\},
\end{equation}
where $g_s$ is the channel power between the SR at the origin and the corresponding ST $Y_o$. From the law of total probability we have
\begin{equation} \label{Eq:PoutSecondary2}
\Pout^{(s)} = \frac{\Pr\l\{\frac{g_s P_s d_s^{-\alpha}}{I_p+I_{s} + \sigma^2}<\theta_s\r\} - \Pr\l\{\frac{g_s P_s d_s^{-\alpha}}{I_p+I_{s} + \sigma^2} <\theta_s \l| \bar{\cE} \r.\r\}\Pr\{\bar{\cE}\}}{\Pr\{\cE\}}.
\end{equation}
Note that $\bar{\cE} = \{\Phi_p \cap b(Y_o,r_g) \neq \emptyset\}$. Then we have the following lemma.

\begin{lemma} \label{Lemma:Approx:PoutSecondary}
Under Assumption~\ref{Assumption:Approximation}, the outage probability of the typical SR at the origin is approximated by
\begin{equation} \label{Approx:PoutSecondary}
\Pout^{(s)} \approx \frac{1-\exp\l(-\tau_s\r)-(1-p_g)}{p_g},
\end{equation}
where
\begin{equation} \label{Eq:Tau_s}
\tau_s = \l(\lambda_p \l(\frac{P_s}{P_p}\r)^{-\frac{2}{\alpha}}+ p_t\lambda_s\r)\theta_s^{\frac{2}{\alpha}}d_s^2\varphi + \frac{\theta_s d_s^\alpha\sigma^2}{P_s}.
\end{equation}
\end{lemma}
\begin{proof} 
See Appendix~\ref{Proof:Lemma:Approx:PoutSecondary}.
\end{proof}

\begin{figure}
\centering
\includegraphics[width=9.5cm]{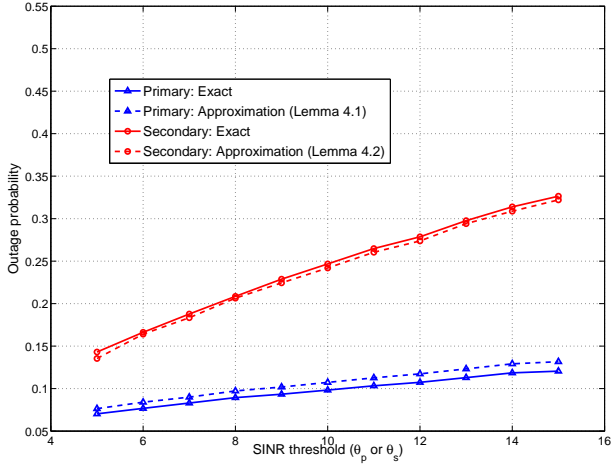}\vspace{5pt} 
\caption{Outage probability of primary and secondary network versus SINR threshold, with $\alpha=4$, $\eta = 0.1$, $d_p=d_s=0.5$, $r_g = 3$, $r_h=1$, $\lambda_p = 0.01$, $\lambda_s = 0.1$, $P_p = 1$, and $P_s = 0.1$.} 
\label{Fig:OutageProb}
\end{figure}

\begin{figure}
\centering
\includegraphics[width=9.5cm]{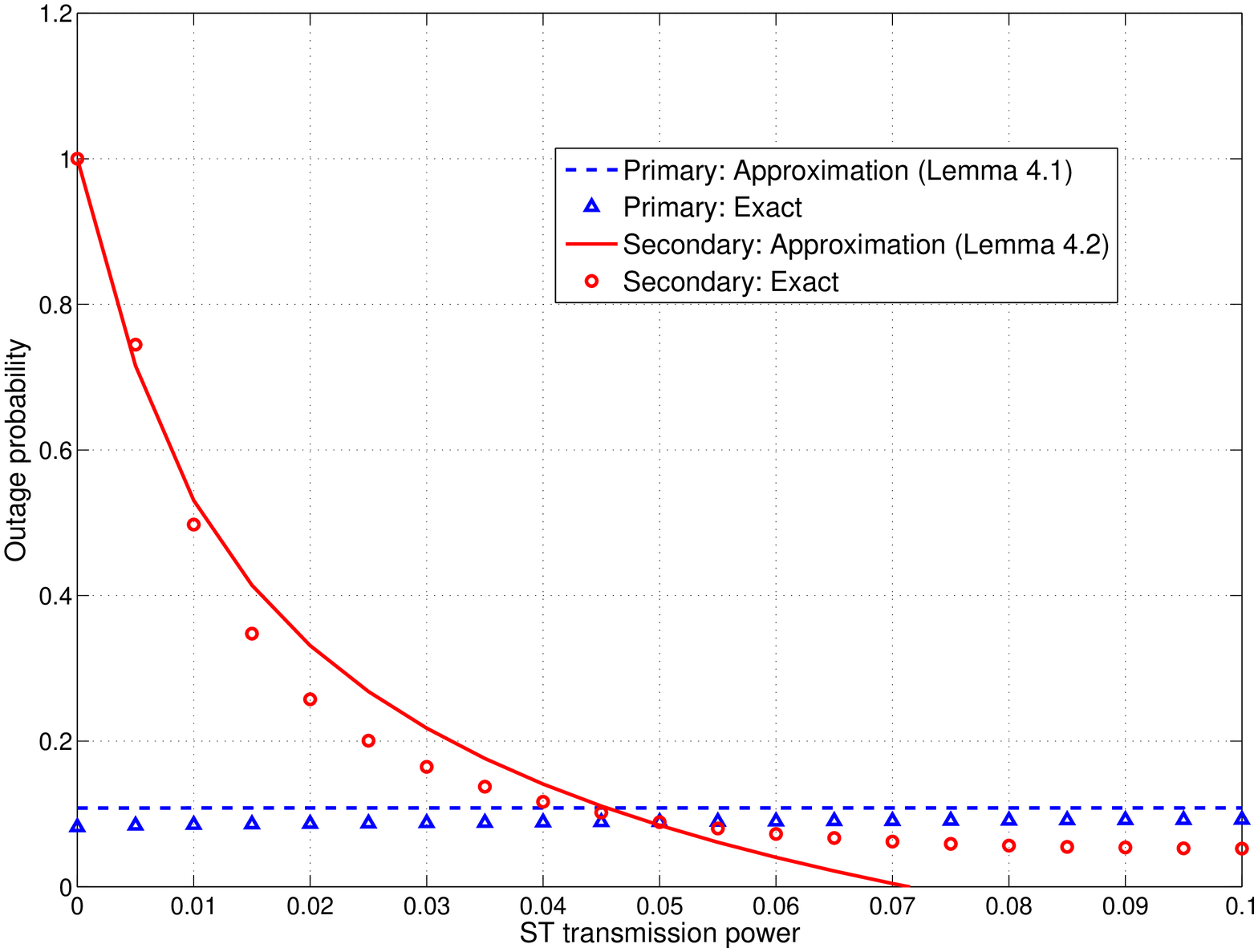}\vspace{5pt} 
\caption{Outage probability of primary and secondary network versus ST transmission power $P_s$, with $\alpha=4$, $\eta = 0.1$, $d_p=d_s=0.5$, $r_g = 4$, $r_h=1$, $\lambda_s = 0.2$, $\lambda_p = 0.01$, $\theta_p=\theta_s=5$, and $P_p = 2$.} 
\label{Fig:OutageProb_vs_P_s}
\end{figure}

Although $I_s$ can be well approximated by \eqref{Eq:Is_Approximation} based on Assumption~\ref{Assumption:Approximation}, it is worth mentioning that the approximated result of $\Pout^{(p)}$ and $\Pout^{(s)}$ in Lemmas~\ref{Lemma:Approx:PoutPrimary2}  and \ref{Lemma:Approx:PoutSecondary}, respectively, are valid only when $P_p \gg P_s$, as assumed in this paper for the following reasons. First, to derive $\Pout^{(p)}$ under Assumption~\ref{Assumption:Approximation}, STs are uniformly located and thus can be inside the guard zone corresponding to the typical PR at the origin, and as a result cause interference to the PR.  However, if we assume $P_p \gg P_s$, the interference due to STs inside this guard zone is negligible and thus can be ignored. Next, to derive $\Pout^{(s)}$, as shown in Appendix~\ref{Proof:Lemma:Approx:PoutSecondary}, the term $\Pr\l\{\frac{g_s P_s d_s^{-\alpha}}{I_p+I_{s} + \sigma^2} <\theta_s \l| \bar{\cE} \r.\r\}$ in \eqref{Eq:PoutSecondary2} can be assumed to be $1$ only when $P_p \gg P_s$. In Figs.~\ref{Fig:OutageProb} and \ref{Fig:OutageProb_vs_P_s}, we compare the outage probabilities obtained by simulations and those based on the  approximations in \eqref{Approx:PoutPrimary2} and \eqref{Approx:PoutSecondary}. It is observed that our approximations are quite accurate and thus Assumption~\ref{Assumption:Approximation} is validated.

In addition, it can be inferred from \eqref{Approx:PoutSecondary} and \eqref{Eq:Tau_s} that $\Pout^{(s)}$ is in general a decreasing function of $P_s$, since $\tau_s$ is a decreasing function of $P_s$. This implies that large ST transmission power $P_s$ is beneficial to reducing  the secondary network outage probability, although larger $P_s$ also increases the interference level from other active STs. This can be explained by the fact that  if $P_s$ is increased, the increase of  received signal power by the SR at the origin can be shown to be more significant than the increase of  interference power from all other active STs. On the other hand, from \eqref{Approx:PoutPrimary2} and \eqref{Eq:Tau_p}, it is analytically difficult to show whether $\Pout^{(p)}$ is a decreasing or increasing function of $P_s$. This is because in general there is a trade-off for setting $P_s$ to minimize the primary outage probability, since larger $P_s$ increases the interference level from active STs (resulting in larger $\Pout^{(p)}$) but at the same time reduces the ST transmission probability $p_t$ (see Fig.~\ref{Fig:TxProb}) and thus the number of active STs (resulting in smaller $\Pout^{(p)}$). In Fig.~\ref{Fig:OutageProb_vs_P_s}, we show the outage probabilities $\Pout^{(p)}$ and $\Pout^{(s)}$ versus $P_s$, respectively. It is observed that $\Pout^{(s)}$ is a decreasing function of $P_s$, whereas $\Pout^{(p)}$ is quite insensitive to the change of $P_s$.

\section{Network Throughput Maximization} \label{Sec:NetThroughput}
In this section, the spatial throughput of the secondary network defined in \eqref{Eq:NetThroughput} is investigated under  the primary and secondary outage constraints. To be more specific, with fixed $P_p$, $\lambda_p$, $r_g$, and $r_h$, the throughput of the secondary network $\cC_s$ is maximized over  $P_s$ and $\lambda_s$ under given $\epsilon_p$ and $\epsilon_s$. The optimization problem can thus be formulated as follows. 
\begin{align}
\label{max}\mathrm{(P1)}:~\mathop{\mathtt{max.}}_{P_s, \lambda_s} &~~   p_t\lambda_s \log_2(1+\theta_s)\\
\mathtt{s.t.}&~~  \Pout^{(p)} \leq \epsilon_p \\
&~~ \Pout^{(s)} \leq \epsilon_s,
\end{align}
where $\Pout^{(p)}$ and $\Pout^{(s)}$ are given by \eqref{Approx:PoutPrimary2} and \eqref{Approx:PoutSecondary}, respectively. With other parameters being fixed, the transmission probability $p_t$ is in general a function of $P_s$ (cf. Section~\ref{Sec:TxProb}). Thus, we denote $p_t$ as $p_t(P_s)$ in the sequel. 

Since $\log_2(1+\theta_s)$ in \eqref{max} is a constant and $\Pout^{(p)}$, $\Pout^{(s)}$ are monotonically increasing functions of $\tau_p$ and $\tau_s$, respectively (see \eqref{Approx:PoutPrimary2} and \eqref{Approx:PoutSecondary}), (P1) is equivalently expressed as 
\begin{align}
 \label{max2}\mathop{\mathtt{max.}}_{P_s, \lambda_s} &~~  p_t(P_s)\lambda_s \\
\mathtt{s.t.}&~~  \tau_p  \leq \mu_p  \label{Ineq:Tau_p} \\  
&~~  \tau_s \leq \mu_s, \label{Ineq:Tau_s}
\end{align}
where $\mu_p = -\ln(1-\epsilon_p)$ and $\mu_s = -\ln((1-\epsilon_s)p_g)$. Note that $\mu_p$ and $\mu_s$ are increasing functions of $\epsilon_p$ and $\epsilon_s$, respectively. In general, it is challenging to find a closed-form solution for \eqref{max2} with $\sigma^2>0$. However, if we assume that the network is primarily interference-limited, by setting $\sigma^2 = 0$, a closed-form solution for (P1) can be obtained as given in the following theorem.

\begin{theorem} \label{Prop:MaxNetThroughput:WithGuardZone}
Assuming $\sigma^2 = 0$, the maximum throughput of the secondary network is given by
\begin{equation} \label{Eq:MaxNetThroughput:WithGuardZone}
\cC_s^* = \frac{\mu_s(\mu_p-\phi\theta_p^{\frac{2}{\alpha}}d_p^2\lambda_p)}{\theta_s^{\frac{2}{\alpha}}d_s^2 \mu_p\phi} \log_2 (1+\theta_s),
\end{equation}
where the optimal ST transmit power is
\begin{equation} \label{Eq:OptTxPower:WithGuardZone}
P_s^* = \frac{\theta_s}{\theta_p}\l(\frac{d_s}{d_p}\r)^\alpha \l(\frac{\mu_s}{\mu_p}\r)^{-\frac{\alpha}{2}}P_p,
\end{equation}
and the optimal ST density is
\begin{equation} \label{Eq:OptPair:WithGuardZone}
\lambda_s^* = \frac{\mu_s(\mu_p-\phi\theta_p^{\frac{2}{\alpha}}d_p^2\lambda_p)}{p_t(P_s^*)\theta_s^{\frac{2}{\alpha}}d_s^2 \mu_p\phi}.
\end{equation}
\end{theorem} 
\begin{proof} 
See Appendix~\ref{Proof:Prop:MaxNetThroughput:WithGuardZone}.
\end{proof}

Note that since $p_t(P_s^*)$ has been obtained in close-form for the case of $0<P_s^* \leq 2 \eta P_p r_h^{-\alpha}$ (i.e., $M=1 \;\text{or}\; M=2$ in Section~\ref{Sec:TxProb}), the optimal ST density $\lambda_s^*$  in \eqref{Eq:OptPair:WithGuardZone} can be obtained exactly for this case, according to \eqref{Eq:TxProb:Case1} and \eqref{Eq:TxProb:Case2}. Otherwise, only upper and lower bounds on $\lambda_s^*$ can be obtained, based on \eqref{Bound:TxProb}.

\begin{figure}
\centering
\includegraphics[width=9cm]{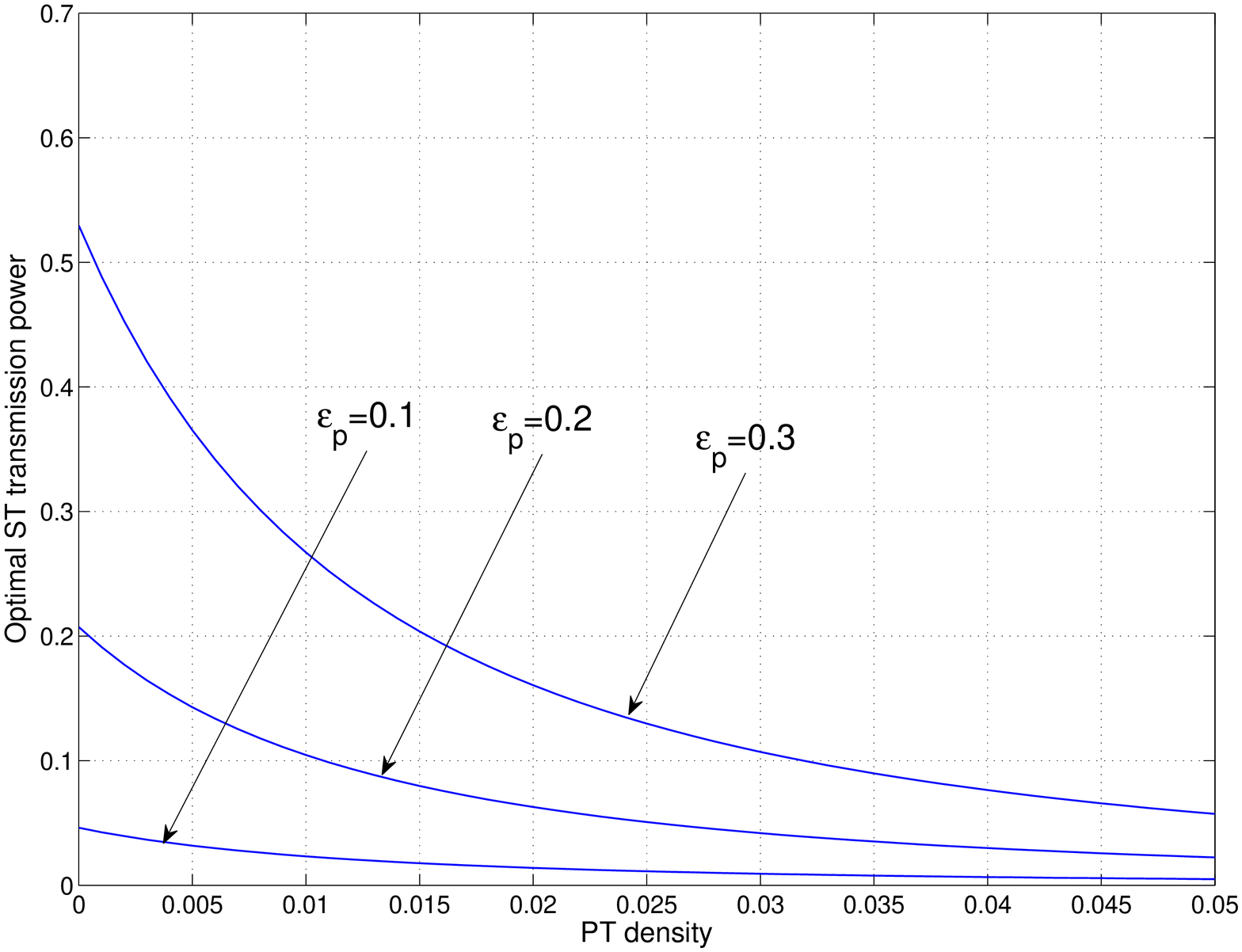}\vspace{5pt} 
\caption{Optimal ST transmission power $P_s^*$ versus PT density $\lambda_p$, with $\alpha=4$, $d_p=d_s=0.5$, $r_h=1$, $r_g=3$, $P_p=2$, $\epsilon_s=0.3$, and $\theta_p=\theta_s=5$.}
\label{Fig:P_s_vs_lambda_p}
\end{figure}

\begin{figure}
\centering
\includegraphics[width=9cm]{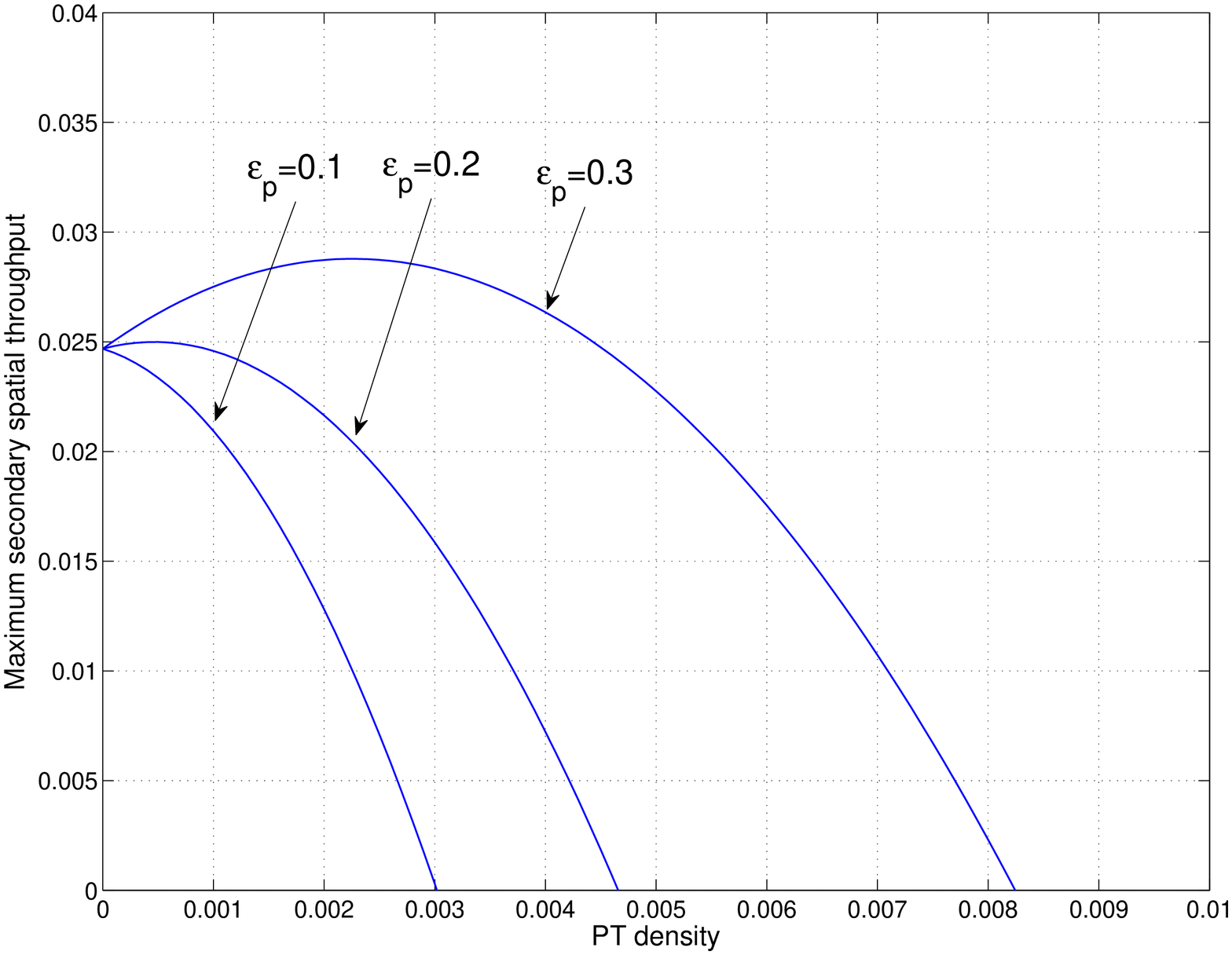}\vspace{5pt} 
\caption{Maximum secondary spatial throughput $\cC_s^*$ versus PT density $\lambda_p$, with $\alpha=4$, $d_p=d_s=0.5$, $r_h=1$, $r_g=3$, $P_p=2$, $\epsilon_s=0.3$, and $\theta_p=\theta_s=5$. }
\label{Fig:C_s_vs_lambda_p}
\end{figure}

\begin{figure}
\centering
\includegraphics[width=9cm]{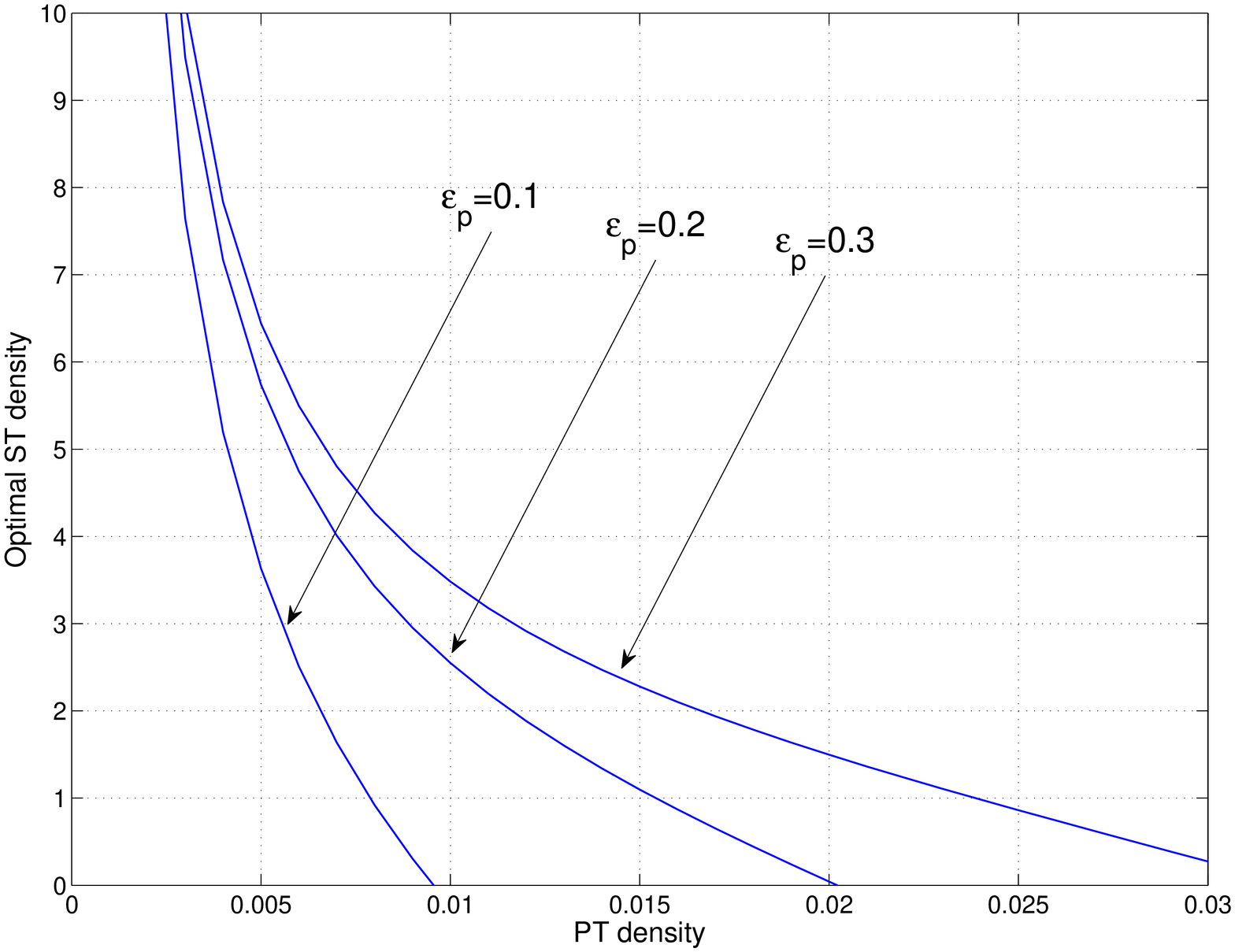}\vspace{5pt} 
\caption{Optimal ST density $\lambda_s^*$ versus PT density $\lambda_p$, with $\alpha=4$, $d_p=d_s=0.5$, $r_h=1$, $r_g=3$, $P_p=2$, $\epsilon_s=0.3$, and $\theta_p=\theta_s=5$. } 
\label{Fig:OptLambda2}
\end{figure}

Some remarks are in order.
\begin{itemize}

\item 
It is worth noting that $\mu_s = -\ln((1-\epsilon_s)p_g)$ in \eqref{Ineq:Tau_s} is an increasing function of PT density $\lambda_p$, since $p_g$ given in \eqref{Eq:GuardZoneRe:a} is a decreasing function of $\lambda_p$. Hence, the optimal ST transmission power $P_s^*$ given in \eqref{Eq:OptTxPower:WithGuardZone} decreases with increasing $\lambda_p$. This result is shown in Fig.~\ref{Fig:P_s_vs_lambda_p}, with three different values of $\epsilon_p$.

\item 
In Fig.~\ref{Fig:C_s_vs_lambda_p}, we show the maximum secondary spatial throughput $\cC_s^*$ given in  \eqref{Eq:MaxNetThroughput:WithGuardZone} versus $\lambda_p$ with $\epsilon_p=0.1$, $0.2$, or $0.3$. Note that from the perspective of RF energy harvesting, larger $\lambda_p$ is beneficial to the secondary network throughput. However, it is observed that if $\epsilon_p=0.1$, $\cC_s^*$ decreases with $\lambda_p$, whlie for $\epsilon_p=0.2$ or $0.3$, $\cC_s^*$ first increases with $\lambda_p$ when $\lambda_p$ is small but eventually starts to decrease when $\lambda_p$ exceeds a certain threshold. The reason of this phenomenon can be explained as follows. When $\epsilon_p$ is small as compared with $\epsilon_s$ (e.g., $\epsilon_p=0.1$ in Fig.~\ref{Fig:C_s_vs_lambda_p}), the constraint in  \eqref{Ineq:Tau_p} prevails over that in \eqref{Ineq:Tau_s}, i.e., satisfying \eqref{Ineq:Tau_p} is sufficient to satisfy \eqref{Ineq:Tau_s}, but not vice versa. Therefore, in this case, if $\lambda_p$ is increased, the active STs' density $p_t\lambda_s$ or $\cC_s^*$ will be decreased to reduce $\tau_p$ in \eqref{Ineq:Tau_p}, i.e.,  reducing the network interference level. However, when $\epsilon_p$ is relatively larger (e.g., $\epsilon_p=0.2$ or $0.3$ in Fig.~\ref{Fig:C_s_vs_lambda_p}), \eqref{Ineq:Tau_s} prevails over \eqref{Ineq:Tau_p}. As a result, if $\lambda_p$ is increased, then so is $\mu_s$ in \eqref{Ineq:Tau_s}, and thus $p_t\lambda_s$ or $\cC_s^*$ will be increased. However, if $\lambda_p$ exceeds a certain threshold, $p_t\lambda_s$ will be decreased to reduce $\tau_s$ in \eqref{Ineq:Tau_s}; as a result, $\cC_s^*$ decreases with increasing $\lambda_p$.

\item It is revealed from \eqref{Eq:OptPair:WithGuardZone} that for given $\lambda_p$, the optimal active STs' density $p_t(P_s^*)\lambda_s^*$ is fixed under a given pair of primary and secondary outage constraints. In other words, $\lambda_s^*$ is inversely proportional to $p_t(P_s^*)$. This implies that as $p_t$ converges to zero with $\lambda_p\rightarrow 0$ (see Fig.~\ref{Fig:TxProbVSPTDen}), $\lambda_s^*$ diverges to infinity at the same time, as shown in Fig.~\ref{Fig:OptLambda2}. Thus, although the sparse PT density will lead to  larger secondary network throughput (see Fig.~\ref{Fig:C_s_vs_lambda_p}), a correspondingly large number of STs  need to be deployed to achieve the maximum throughput, each with a very small transmission probability $p_t$. As a result, only a small fraction of the STs could be active at any time, resulting in large delay for secondary transmissions or inefficient secondary network design.

\end{itemize}

\section{Application and Extension} \label{Sec:Application}
In this section, we extend our results on the CR network to the application scenario depicted in Fig.~\ref{Fig:NetworkModel_wo}, where a set of distributed wireless power chargers (WPCs) are deployed to power wireless information transmitters (WITs) in a sensor network. It is assumed that wireless power transmission from WPCs to WITs is over a dedicated band which is different from that for the information transfer, and thus does not interfere with wireless information receivers (WIRs). For simplicity, we assume that the path-loss exponents for both the power transmission and information transmission are equal to $\alpha$. Moreover, the network models for WPCs and WITs as well as the energy harvesting and transmission models of WITs are similarly assumed as in Section~\ref{Sec:SysMod} for PTs and STs in the CR setup. For convenience, we thus use the same symbol notations for PTs and STs to represent for WPCs and WITs, respectively.

\subsection{Transmission Probability}
Unlike the CR case, WITs in a sensor network do not need to be prevented from  transmissions by guard zones, since there are no PTs present. As a result, a WIT can transmit at any time provided that it is fully charged. By letting $r_g = 0$, we have $p_g = 1$, and from \eqref{Eq:TxProb:Case1}, \eqref{Eq:TxProb:Case2} and \eqref{Bound:TxProb} we obtain the transmission probability of a typical WIT in the following corollary. 
\begin{corollary} \label{Corollary:TxProb:WithoutGuardZone}
The transmission probability of a typical WIT is given by
\begin{enumerate}
\item If $0<P_s \leq \eta P_p r_h^{-\alpha}$ or $M=1$,
\begin{equation}
p_t = \frac{p_h}{1+p_h}.
\end{equation}

\item If $\eta P_p r_h^{-\alpha} < P_s \leq 2\eta P_p  r_h^{-\alpha}$ or $M=2$,
\begin{equation}
p_t = \frac{p_h}{p_h + 1 + \frac{p_{2}}{p_h}}.
\end{equation}

\item If $P_s > 2\eta P_p  r_h^{-\alpha}$ or $M>2$,
\begin{equation}
\frac{p_1 + p'_2}{p_1+p'_2 + 1 + \frac{p'_2}{p_1+p'_2}} \leq p_t \leq \frac{p_h}{p_h + 1+\frac{p'_2+p_3}{p_h}},
\end{equation}
\end{enumerate}
\end{corollary}
where $p_h=1-e^{-\pi r_h^2 \lambda_p}$ is given in \eqref{Eq:p_h}; $p_1=1-e^{-\pi h_1^2 \lambda_p}$ and  $p_2=e^{-\pi h_1^2 \lambda_p} - e^{-\pi r_h^2 \lambda_p}$ are given in \eqref{Eq:p_1} and \eqref{Eq:p_2}, respectively; $p'_2=e^{-\pi \lambda_p h_1^2} - e^{-\pi \lambda_p h_2^2}$ and $p_3=e^{-\pi \lambda_p h_2^2} - e^{-\pi \lambda_p r_h^2}$ are given in \eqref{Eq:p'_2} and \eqref{Eq:p_3}, respectively.

It is worth noting that unlike the CR setup, $p_t$ in this case is in general an increasing function of $\lambda_p$ since there are no guard zones and thus larger $\lambda_p$ always help charge WITs more frequently, as shown in Fig.~\ref{Fig:TxProbVSPTDen:WithoutGuardZone}. 

\begin{figure}
\centering
\includegraphics[width=9.2cm]{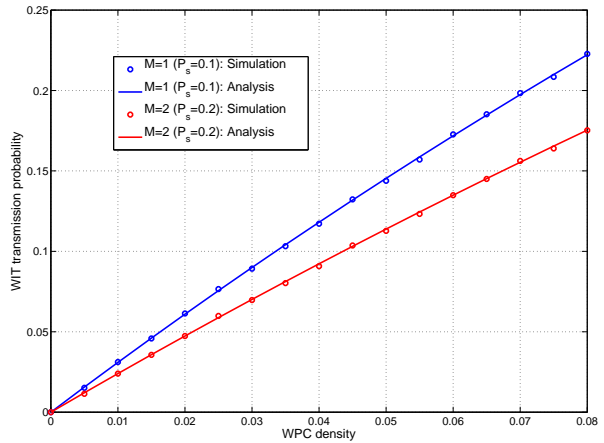}\vspace{5pt} 
\caption{WIT transmission probability $p_t$ versus WPC density $\lambda_p$, with $\alpha=4$, $\eta=0.1$, $r_h=1$, $r_g=3$, and $P_p=1$.}
\label{Fig:TxProbVSPTDen:WithoutGuardZone}
\end{figure}

\subsection{Network Throughput Maximization}
Note that unlike the CR setup, here we only need to consider the outage probability of a typical WIR at the origin due to the interference of other active WITs. Similar to Assumption~\ref{Assumption:Approximation}, we assume that active WITs form an HPPP with density $p_t\lambda_s$; thus, the outage probability of a typical WIR at the origin can be obtained by simplifying Lemma~\ref{Lemma:Approx:PoutPrimary2} as
\begin{align} \label{Eq:PoutSecondary:WithoutGuardZone}
\Pout^{(s)} & = \Pr\l\{\frac{g_s P_s d_s^{-\alpha}}{I_s + \sigma^2} < \theta_s\r\} \\
&=  1-\exp \l(-\tau_s \r), \label{Eq:PoutIT}
\end{align}
where in this case $\tau_s$ is given by
\begin{equation} \label{Eq:TildeTau_s}
\tau_s = \theta_s^{\frac{2}{\alpha}}d_s^2\varphi \,p_t\lambda_s + \frac{\theta_s d_s^\alpha \sigma^2}{P_s}.
\end{equation}

For the sensor network throughput maximization, Problem (P1) can be modified such that only the outage constraint for the WIR  is applied. Thus we have the following simplified problem.
\begin{align}
\mathrm{(P2)}:~\mathop{\mathtt{max.}}_{P_s, \lambda_s} &~~   p_t(P_s)\lambda_s \log_2(1+\theta_s)\\
\mathtt{s.t.}&~~ \Pout^{(s)} \leq \epsilon_s.
\end{align}
The solution of (P2) is given in the following corollary, based on Theorem~\ref{Prop:MaxNetThroughput:WithGuardZone}. 

\begin{corollary} \label{Prop:MaxNetThroughput:WithoutGuardZone}
Assuming $\sigma^2 = 0$, the maximum network throughput is given by
\begin{equation} \label{Eq:MaxNetThroughput:WithoutGuardZone}
\cC_s^* =  \frac{\mu'_s}{\theta_s^{\frac{2}{\alpha}}d_s^2\varphi}\log_2 (1+\theta_s),
\end{equation}
where $\mu'_s = -\ln(1-\epsilon_s)$, and the optimal solution $(P_s^*,\lambda_s^*) \in \mR_+\times \mR_+$ is any pair satisfying
\begin{equation} \label{Eq:OptPoint:WithoutGuardZone}
p_t(P_s^*)\lambda_s^* = \frac{\mu'_s}{\theta_s^{\frac{2}{\alpha}}d_s^2\varphi}.
\end{equation}
\end{corollary}

\begin{proof}
With $\sigma^2 = 0$, from \eqref{Eq:PoutIT} and \eqref{Eq:TildeTau_s}, Problem (P2) can be equivalently rewritten as 
\begin{align}
\max_{P_s, \lambda_s} & \quad  p_t(P_s) \lambda_s\\
\mbox{s.t.}& \quad p_t(P_s)\lambda_s \leq \frac{\mu'_s}{\theta_s^{\frac{2}{\alpha}}d_s^2\varphi}, \label{Eq:P_2}
\end{align}
where $\mu'_s = -\ln(1-\epsilon_s)$. To maximize $p_t(P_s)\lambda_s$, then it is easy to see from \eqref{Eq:P_2} that the optimal solution is $p_t(P_s^*)\lambda_s^*=\frac{\mu'_s}{\theta_s^{\frac{2}{\alpha}}d_s^2\varphi}$; by multiplying it with $\log_2(1+\theta_s)$, we then obtain $\cC_s^*$ in \eqref{Eq:MaxNetThroughput:WithoutGuardZone}.

\end{proof}
Note that unlike the result in Theorem~\ref{Prop:MaxNetThroughput:WithGuardZone}, the maximum network throughput remains constant regardless of $\lambda_p$. This is because there is no primary outage constraint in this case and thus the optimal density of active WITs $p_t(P_s^*)\lambda_s^*$ is determined solely by the outage constraint of WIRs. On the other hand, if $\lambda_p$ is increased, we can effectively reduce the required WIT density $\lambda_s^*$ for achieving the same $\cC_s^*$ since $p_t$ in general  increases with $\lambda_p$.

\section{Conclusion} \label{Sec:Con}
In this paper, we have proposed a novel network architecture enabling secondary users to harvest energy as well as  reuse the spectrum of primary users in the CR network. Based on stochastic-geometry models and certain assumptions, our study revealed useful insights to optimally design the RF energy powered CR network. We derived the transmission probability of a secondary transmitter by considering the effects of both the guard zones and harvesting zones, and thereby characterized the maximum secondary network throughput under the given outage constrains for primary and secondary users, and the corresponding optimal secondary transmit power and transmitter density in closed-form. Moreover, we showed that our result can also be applied to the wireless sensor network powered by a distributed WPC network, or other similar wireless powered communication networks.

\appendices 
\section{Proof of Proposition~\ref{Theorem:Bound:TxProb}} \label{Proof:Theorem:Bound:TxProb}
For both the upper and lower bounds, similar to the case of $M=2$, we apply a 3-state Markov chain with state space $\{0,1,2\}$ with states $0$, $1$ and $2$ denoting the battery power level of $0$, in the range $[\frac{1}{2}P_s, P_s)$, and equal to $P_s$, respectively. First, consider the upper bound on $p_t$. Since the harvested power in the region $a(X,h_2,r_h)$ is assumed to be equal to $\frac{1}{2}P_s$, it is easy to see that the state transition-probability matrix for this case is given by 
\begin{equation} \label{Eq:Matrix:Upper}
\mathbf{P}^{(u)}=
\l[\begin{array}{ccc}
1-p_h & p'_2 + p_3 & p_{1}\\0 & 1-p_h & p_h  \\ p_g & 0 & 1-p_g\end{array} \r]. 
\end{equation}
Let $\boldsymbol{\pi}^{(u)}=[\pi^{(u)}_{0}, \pi^{(u)}_{1}, \pi^{(u)}_{2}]$ denote the steady-state probability vector in this case. Solving $\boldsymbol{\pi}^{(u)}\mathbf{P}^{(u)} = \boldsymbol{\pi}^{(u)}$, we obtain $\pi^{(u)}_{2} = \frac{p_h}{p_h + p_g\l(1+\frac{p'_2 + p_3}{p_h}\r)}$ and thus the upper bound on $p_t$ can be obtained by multiplying $\pi^{(u)}_{2}$ with $p_g$, according to \eqref{Eq:TxProb:Def}.

Next, consider the lower bound on $p_t$. Since the harvested power in the region $a(X,h_2,r_h)$ is assumed to be $0$, it is easy to obtain the state transition-probability matrix for this case as
\begin{equation} \label{Eq:Matrix:Lower}
\mathbf{P}^{(l)}=
\l[\begin{array}{ccc}
1-(p_1+p'_2) & p'_2& p_{1}\\0 & 1-(p_1+p'_2) & p_1+p'_2  \\ p_g & 0 & 1-p_g\end{array} \r]. 
\end{equation}
Similarly to the derivation of the upper bound on $p_t$, the lower bound on $p_t$ can be found by finding the corresponding steady-state probability $\pi^{(l)}_{2} = \frac{p_1+p'_2}{(p_1+p'_2) + p_g\l(1+\frac{p'_2}{p_1+p'_2}\r)}$, and then multiplying it with $p_g$. The proof of Proposition~\ref{Theorem:Bound:TxProb} is thus completed.

\section{Proof of Lemma~\ref{Lemma:Approx:PoutPrimary2}} \label{Proof:Lemma:Approx:PoutPrimary2}
For convenience, we derive the non-outage probability $1-\Pout^{(p)}$ as follows with $\Pout^{(p)}$ given in \eqref{Eq:PoutPrimary}.
\begin{align}
1-\Pout^{(p)} & = \Pr\l\{\frac{g_p P_p d_p^{-\alpha}}{I_p + I_s + \sigma^2} \geq \theta_p\r\}  \label{Approx:PoutPrimaryProof}\\
& = \Pr\l\{g_p \geq \frac{\theta_p d_p^\alpha}{P_p}\l(I_p + I_s + \sigma^2\r)\r\} \\
& = \E_{I_p}\l[\E_{I_s}\l[\exp\l({-\frac{\theta_p d_p^\alpha}{P_p}\l(I_p + I_s + \sigma^2\r)}\r)\r]\r] \\
& = \exp\l({-\frac{\theta_p d_p^\alpha}{P_p}\sigma^2}\r)\E_{I_p}\l[\exp\l({-\frac{\theta_p d_p^\alpha}{P_p}I_p}\r)\r]\nn\\&\qquad\qquad\qquad\E_{I_s}\l[\exp\l({-\frac{\theta_p d_p^\alpha}{P_p} I_s }\r)\r],  \label{Approx:PoutPrimaryProof2}
\end{align}
where in \eqref{Approx:PoutPrimaryProof2}, the expectations are separated since $I_p$ and $I_s$ are assumed to be independent as a result of  Assumption~\ref{Assumption:Approximation}. Note that $\E_{I_p}\l[\exp\l({-\frac{\theta_p d_p^\alpha}{P_p}I_p}\r)\r]$ and $\E_{I_s}\l[\exp\l({-\frac{\theta_p d_p^\alpha}{P_p} I_s}\r)\r]$ are Laplace transforms in terms of the random variables $I_p$ and $I_s$, respectively, both with input parameter $\frac{\theta_p d_p^\alpha}{P_p}$. According to the result in \cite[3.21]{Haenggi:InterferenceLargeWirelessNetworks:2008}, the Laplace transform of the shot-noise process of an HPPP $\Lambda(\lambda)$ with density $\lambda>0$, denoted by $I = \sum_{T\in\Lambda(\lambda)}g_T P |T|^{-\alpha}$, with input parameter $s$ is given by
\begin{equation} \label{Eq:Laplace}
\E_I[\exp(-sI)] = \exp(-(Ps)^{\frac{2}{\alpha}} \lambda \phi),
\end{equation}
where $\{g_T\}_{T\in\Lambda(\lambda)}$ is a set of i.i.d. exponential random variables with mean $1$, and $\phi$ is given in Lemma~\ref{Lemma:Approx:PoutPrimary2}.
Using \eqref{Eq:Laplace}, we can easily obtain $\E_{I_p}\l[\exp\l({-\frac{\theta_p d_p^\alpha}{P_p}I_p}\r)\r]$ and $\E_{I_s}\l[\exp\l({-\frac{\theta_p d_p^\alpha}{P_p} I_s}\r)\r]$ and by substituting them to \eqref{Approx:PoutPrimaryProof2}, the proof of Lemma~\ref{Lemma:Approx:PoutPrimary2} is thus completed.

\section{Proof of Lemma~\ref{Lemma:Approx:PoutSecondary}} \label{Proof:Lemma:Approx:PoutSecondary}
The term $\Pr\l\{\frac{g_s P_s d_s^{-\alpha}}{I_p+I_{s} + \sigma^2}<\theta_s\r\}$ in \eqref{Eq:PoutSecondary2} is obtained by following the similar procedure in the proof of Lemma~\ref{Lemma:Approx:PoutPrimary2}, given by 
\begin{equation} \label{Eq:Superposition2}
\Pr\l\{\frac{g_s P_s d_s^{-\alpha}}{I_p+I_s + \sigma^2}<\theta_s\r\} = 1-\exp(\tau_s),
\end{equation}
where $\tau_s$ is given in \eqref{Eq:Tau_s}.

Next, under the assumption $P_p \gg P_s$, it is reasonable to assume that the interference from even only one single PT inside $b(Y_o, r_g)$ is sufficient to cause an outage to the typical SR at the origin. Consequently, we have $\Pr\l\{\frac{g_s P_s d_s^{-\alpha}}{I_p+I_{s} + \sigma^2} <\theta_s \l| \bar{\cE} \r.\r\} \approx 1$. Substituting this result, \eqref{Eq:Superposition2} and $\Pr\{\cE\}=e^{-\pi r_g^2 \lambda_p}=1-p_g$ into \eqref{Eq:PoutSecondary2}  yields   \eqref{Approx:PoutSecondary}. The proof of Lemma~\ref{Lemma:Approx:PoutSecondary} is thus completed.

\begin{figure}
\centering
\includegraphics[width=9cm]{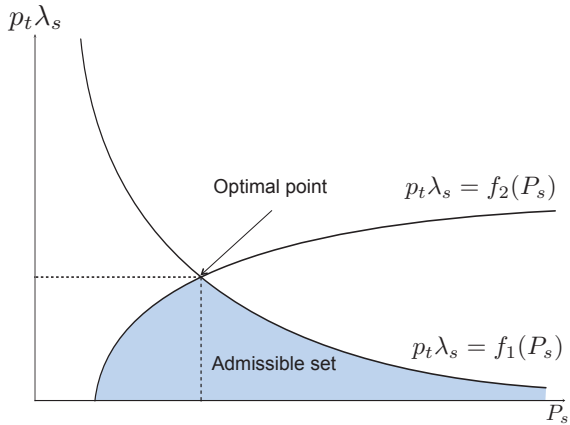}\vspace{5pt} 
\caption{Illustration of the optimal solution for Problem (P1).}\label{Fig:AdmissibleSet}
\end{figure}

\section{Proof of Theorem~\ref{Prop:MaxNetThroughput:WithGuardZone}} \label{Proof:Prop:MaxNetThroughput:WithGuardZone}
From \eqref{Eq:Tau_p} and \eqref{Eq:Tau_s}, the constraints $\tau_p \leq \mu_p$ and $\tau_s \leq \mu_s$ given in \eqref{Ineq:Tau_p} and \eqref{Ineq:Tau_s} are equivalent to $p_t(P_s)\lambda_s \leq f_1(P_s)$ and $p_t(P_s)\lambda_s \leq f_2(P_s)$, respectively, where
\begin{equation}\label{Eq:f_1} \vspace{-5pt}
f_1(P_s) =\l[\frac{1}{\theta_p^{\frac{2}{\alpha}}d_p^2\varphi}\l(\mu_p - \frac{\theta_p d_p^\alpha \sigma^2}{P_p}\r) - \lambda_p\r] \l(\frac{P_s}{P_p}\r)^{-\frac{2}{\alpha}},
\end{equation}
\begin{equation}
f_2(P_s) =  \frac{1}{\theta_s^{\frac{2}{\alpha}}d_s^2\varphi}\l(\mu_s - \frac{\theta_s d_s^\alpha \sigma^2}{P_s}\r) - \lambda_p \l(\frac{P_s}{P_p}\r)^{-\frac{2}{\alpha}}.
\end{equation}
As illustrated in Fig.~\ref{Fig:AdmissibleSet}, $f_1(P_s)$ decreases whereas $f_2(P_s)$ increases with growing $P_s$. The shaded region in Fig.~\ref{Fig:AdmissibleSet} shows the admissible set of $(P_s, p_t\lambda_s)$ that satisfies the given outage probability constraints. It is observed that the optimal value of $p_t(P_s)\lambda_s$ is the intersection of the two curves $p_t(P_s)\lambda_s = f_1(P_s)$ and $p_t(P_s)\lambda_s = f_2(P_s)$. The intersection point can be found by solving $f_1(P_s) = f_2(P_s)$, which has no closed-form solution in general with $\sigma^2 >0$. However, by letting $\sigma^2=0$, the closed-form solution of $P_s^*$ can be obtained as $\frac{\theta_s}{\theta_p}\l(\frac{d_s}{d_p}\r)^\alpha \l(\frac{\mu_s}{\mu_p}\r)^{-\frac{\alpha}{2}}P_p$. From $p_t(P_s^*)\lambda_s^* = f_1(P_s^*)$ and \eqref{Eq:f_1}, we then obtain $p_t(P_s^*)\lambda_s^* = \frac{\mu_s(\mu_p-\phi\theta_p^{\frac{2}{\alpha}}d_p^2\lambda_p)}{\theta_s^{\frac{2}{\alpha}}d_s^2 \mu_p\phi}$, and accordingly $\lambda_s^* = \frac{\mu_s(\mu_p-\phi\theta_p^{\frac{2}{\alpha}}d_p^2\lambda_p)}{p_t(P_s^*)\theta_s^{\frac{2}{\alpha}}d_s^2 \mu_p\phi}$. Theorem~\ref{Prop:MaxNetThroughput:WithGuardZone} is thus proved.

\bibliographystyle{ieeetr}

\end{document}

%% file: header.tex
\def\E{\mathsf{E}}

\def\phi{\varphi}

\def\SINR{\mathsf{SINR}}

\def\l{\left}
\def\r{\right}
\def\({\left(}
\def\){\right)}

\def\b0{{\mathbf{0}}}

\def\mR{{\mathbb{R}}}

\def\cA{\mathcal{A}}

\def\cC{\mathcal{C}}

\def\cE{\mathcal{E}}

\def\cG{\mathcal{G}}
\def\cH{\mathcal{H}}

\newcommand{\Pout}{P_{\mathsf{out}}}

\newcommand{\nn}{\nonumber}